\documentclass[letterpaper,11pt]{article}
\usepackage{microtype}
\usepackage[letterpaper,margin=0.97in]{geometry}
\usepackage[numbers,sort&compress]{natbib}
\usepackage[table]{xcolor}
\usepackage{amsmath}
\usepackage{amssymb}
\usepackage{mathrsfs}
\usepackage{amsthm}
\usepackage{enumitem}
\usepackage{textgreek}
\usepackage{graphicx}
\usepackage{framed}
\usepackage[small]{caption}
\usepackage{booktabs}
\usepackage{tikz-cd}
\usepackage{hyperref}
\usepackage{doi}
\usepackage[textsize=tiny]{todonotes}
\usepackage[framemethod=tikz]{mdframed}
\usepackage{thm-restate}
\usepackage{mathtools}
\usepackage{xspace}
\usepackage[capitalize, nameinlink]{cleveref}
\usepackage{multirow}
\usepackage{cellspace}
\usepackage{array}
\usepackage{tabularx}
\usepackage[framemethod=tikz]{mdframed}

\Crefname{remark}{Remark}{Remarks}
\Crefname{observation}{Observation}{Observations}

\theoremstyle{plain}
\newtheorem{theorem}{Theorem}[section]
\newtheorem{lemma}[theorem]{Lemma}

\newtheorem{corollary}[theorem]{Corollary}

\theoremstyle{definition}
\newtheorem{definition}[theorem]{Definition}

\theoremstyle{plain}

\theoremstyle{remark}

\newcommand{\namedref}[2]{\hyperref[#2]{#1~\ref*{#2}}}

\newcommand{\eps}{\varepsilon}
\DeclareMathOperator{\poly}{poly}
\newcommand{\LOCAL}{\ensuremath{\mathsf{LOCAL}}\xspace}
\newcommand{\CONGEST}{\ensuremath{\mathsf{CONGEST}}\xspace}
\newcommand{\set}[1]{\left\{#1\right\}}
\newcommand{\degmin}{d_{\phi}^{-}}
\newcommand{\phalpha}[1]{\alpha_{#1}(\phi)}

\newenvironment{myabstract}
{\list{}{\listparindent 1.5em%
		\itemindent    \listparindent
		\leftmargin    1cm
		\rightmargin   1cm
		\parsep        0pt}%
	\item\relax}
{\endlist}

\newenvironment{mycover}
{\list{}{\listparindent 0pt
		\itemindent    \listparindent
		\leftmargin    1cm
		\rightmargin   1cm
		\parsep        0pt}%
	\raggedright
	\item\relax}
{\endlist}

\newcommand{\myemail}[1]{\,$\cdot$\, {\small #1}}
\newcommand{\myaff}[1]{\,$\cdot$\, {\small #1}\par\smallskip}

\definecolor{darkgreen}{rgb}{0,0.5,0}
\definecolor{darkred}{rgb}{0.4,0,0}
\hypersetup{
	colorlinks=true,
	linkcolor=darkred,
	citecolor=darkgreen,
	filecolor=black,
	urlcolor=[rgb]{0,0.1,0.5},
	pdftitle={Distributed Edge Coloring in Time Polylogarithmic in $\Delta$},
	pdfauthor={Alkida Balliu, Sebastian Brandt, Fabian Kuhn, Dennis Olivetti}
}

\begin{document}

	\begin{mycover}
		{\huge\bfseries\boldmath Distributed Edge Coloring in Time Polylogarithmic in $\Delta$ \par}
		\bigskip
		\bigskip
		\bigskip
		
		\textbf{Alkida Balliu}
		\myemail{alkida.balliu@gssi.it}
		\myaff{Gran Sasso Science Institute}
		
		\textbf{Sebastian Brandt}
		\myemail{brandt@cispa.de}
		\myaff{CISPA Helmholtz Center for Information Security}
		
		\textbf{Fabian Kuhn}
		\myemail{kuhn@cs.uni-freiburg.de}
		\myaff{University of Freiburg}
		
		\textbf{Dennis Olivetti}
		\myemail{dennis.olivetti@gssi.it}
		\myaff{Gran Sasso Science Institute}
	\end{mycover}
	\bigskip

\begin{myabstract}
  We provide new deterministic algorithms for the edge coloring problem, which is one of the classic and highly studied distributed local symmetry breaking problems. As our main result, we show that a $(2\Delta-1)$-edge coloring can be computed in time $\poly\log\Delta + O(\log^* n)$ in the \LOCAL model. This improves a result of Balliu, Kuhn, and Olivetti [PODC '20], who gave an algorithm with a quasi-polylogarithmic dependency on $\Delta$. We further show that in the \CONGEST model, an $(8+\eps)\Delta$-edge coloring can be computed in $\poly\log\Delta + O(\log^* n)$ rounds. The best previous $O(\Delta)$-edge coloring algorithm that can be implemented in the \CONGEST model is by Barenboim and Elkin [PODC '11] and it computes a $2^{O(1/\eps)}\Delta$-edge coloring in time $O(\Delta^\eps + \log^* n)$ for any $\eps\in(0,1]$. 
\end{myabstract}

\clearpage
\setcounter{page}{0}
\thispagestyle{empty}
\tableofcontents

\clearpage

\section{Introduction}

In the most standard setting of distributed graph algorithms, we are
given a network that is modeled as an undirected graph $G=(V,E)$. The
nodes $V$ are the active entities of the network and communication
happens by exchanging messages over the edges $E$ in synchronous
rounds. The goal is to solve some graph problem on $G$. This
distributed computation model is known as the \LOCAL model if the
communication between neighbors in each round is not restricted and it
is known as the \CONGEST model if the messages exchanged neighbors
have to consist of at most $O(\log n)$ bits (where $n=|V|$)~\cite{linial1987,Peleg2000}. Four graph
problems that have received particular attention in this context are
the problems of computing a maximal independent set (MIS) of $G$, a $(\Delta+1)$-vertex coloring of $G$ (where
$\Delta$ is the maximum degree of $G$),  a
maximal matching of $G$, and a
$(2\Delta-1)$-edge coloring of $G$. All four problems have in common
that they can be solved by a trivial sequential greedy algorithm.
The four problems can be seen as prototypical examples of distributed
symmetry breaking problems and understanding the distributed
complexities of them has been at the very core of the
area of distributed graph algorithms, e.g.\
\cite{Barenboim2013,linial1987,stoc17_complexity}. All four problems
can be solved by quite simple and straightforward $O(\log n)$-round
randomized distributed algorithms, which have been known for more than
thirty years~\cite{Alon1986,Israeli1986,linial1987,luby1986}. In light
of those efficient randomized algorithms, a lot of the work
concentrated on developing deterministic distributed algorithms for
the four problems. As the contributions of the present paper are on
deterministic algorithms, we also focus on deterministic algorithms
when discussing the prior work in the following.

\paragraph{Deterministic Complexity as a Function of the Network Size.} In
\cite{Awerbuch89,panconesi96decomposition}, a tool called network
decomposition was introduced as a generic technique to solve
distributed graph problems in the \LOCAL model. This resulted in deterministic
$2^{O(\sqrt{\log n})}$-round algorithms in particular for the four
problems discussed above and it left open the question of whether the
problems can also be solved deterministically in polylogarithmic time
and thus similarly fast as with the simple randomized algorithms of
the 1980s. This was first shown for the maximal matching problem in
\cite{Hanckowiak1998,HanckowiakKP99}, where the authors gave a deterministic
$O(\log^4 n)$-time algorithm. The result was later improved to the
current running time of $O(\log^2\Delta\cdot\log n)$ in
\cite{fischer17improved}. Much more recently, it was shown that also the
$(2\Delta-1)$-edge coloring problem can also be solved in polylogarithmic
time deterministically~\cite{FischerGK17,FOCS18-derand,Harris20} by reducing the problem to the
problem of computing a maximal matching in $3$-uniform hypergraphs and
by giving a polylogarithmic-time deterministic distributed algorithm
for maximal matching in hypergraphs of bounded rank. The algorithm of
\cite{Harris20} achieves a round complexity of
$\tilde{O}(\log^2\Delta\cdot\log n)$ and is thus almost as fast as the
maximal matching algorithm of \cite{fischer17improved}. Subsequently, in a
breakthrough result, Rozho\v{n} and Ghaffari~\cite{Rozhon2020} gave a
polylogarithmic-time algorithm for the network decomposition
problem. The current best version of this algorithm in \cite{GGR2020}
implies $O(\log^5 n)$-time deterministic distributed algorithms for
all four problems discussed above. Finally, in a recent
paper~\cite{GhaffariKuhn20}, a direct $O(\log^2\Delta\cdot\log n)$-round
algorithm for the $(\Delta+1)$-vertex coloring (and thus also for the
$(2\Delta-1)$-edge color problem) was presented.

\paragraph{Deterministic Complexity as a Function of the Maximum
  Degree.} The optimal time complexity of a given graph problem in the
\LOCAL model can be interpreted as the \emph{locality} of the problem
in the following sense. If there is an $R$-round deterministic \LOCAL
algorithm for solving a problem in a graph $G$, then every node of
$G$ can compute its output as a function of its $R$-hop neighborhood
 in $G$ and if more than $R$ rounds are needed to solve a
problem, then some node must learn something about the graph that is
outside the node's $R$-hop neighborhood~\cite{linial1987,Peleg2000}. The local neighborhood
of a node is in principle independent of the size of the network. It
is therefore natural to ask for the locality of graph problems not
just as a function of the network size, but also as a function of more
local properties. Further, local graph algorithms are particularly interesting
in large networks where the node degrees might be independent or
almost independent of the network size. As a result, there is an
extended body or prior work that tries to understand the distributed
complexity of graph problems as a function of the maximum degree
$\Delta$ rather than as a function of $n$.

Note however that the distributed complexity of many problems cannot
be completely independent of $n$. Linial~\cite{linial1987} showed that
even in networks of maximum degree $\Delta=2$, computing a coloring
with $O(1)$ colors (and by simple reductions computing solutions for
all four classic problems discussed above) requires at least
$\Omega(\log^* n)$ rounds.\footnote{Note that even in bounded-degree
  graphs, the local neighborhoods are not completely independent of
  the network size. In order to equip the network with unique
  identifiers, the space from which the identifiers are chosen has to
  grow as a function of the network size $n$.} However, in
\cite{linial1987}, it is also shown that in $O(\log^* n)$ rounds, one
can compute a vertex coloring with $O(\Delta^2)$ colors for any graph
$G$. Given such a coloring, one can then iterate through the color
classes and obtain simple distributed $O(\Delta^2)$-round
implementations of the natural sequential greedy algorithms for
computing an MIS, a maximal matching, a $(\Delta+1)$-vertex coloring,
or a $(2\Delta-1)$-edge coloring. All four problems can therefore be
solved in $O(\Delta^2+\log^* n)$ rounds in the \LOCAL model (and also
in the \CONGEST model). That is, we can solve the problems in time
$O(f(\Delta)+\log^*n)$ for some function $f$. This keeps the
dependency on $n$ as small as it can be. When establishing the local
complexity of a distributed graph problem, we are interested in
optimizing the function $f$ and thus the $\Delta$-dependency of this
bound.

Starting from the early work on distributed graph algorithms, there is a long line of research that tries to optimize the $\Delta$-dependency of the four discussed problems~\cite{linial1987,goldberg88,SzegedyV93,Panconesi01simple,KuhnW06,barenboim14distributed,BarenboimE10,BarenboimE11,barenboim16sublinear,fraigniaud16local,BEG18,MausTonoyan20,Kuhn20,BalliuKO20}. As given a $C$-vertex or a $C$-edge coloring, all four problems can be solved in $C$ rounds, the primary focus was on developing efficient coloring algorithms. In a first phase, the time for computing a $(2\Delta-1)$-edge coloring~\cite{Panconesi01simple} and  for computing a $(\Delta+1)$-vertex coloring was improved to $O(\Delta+\log^* n)$~\cite{barenboim14distributed,Kuhn2009}. As a result, we therefore also obtain $O(\Delta+\log^*n)$-round algorithms for MIS and maximal matching. In \cite{Balliu2019,trulytight,balliu2021hideandseek}, it was shown that for MIS and maximal matching, this bound is tight, even on tree networks. More specifically, it was shown that there is no randomized MIS or maximal matching algorithm with round complexity $o\big(\Delta + \frac{\log\log n}{\log\log\log n}\big)$ and there is no deterministic such algorithm with round complexity $o\big(\Delta + \frac{\log n}{\log\log n}\big)$.  While for MIS and matching, there is a linear-in-$\Delta$ lower bound, for $(\Delta+1)$-vertex coloring and $(2\Delta-1)$-edge coloring, there are in fact algorithms with a sublinear-in-$\Delta$ complexity. This was first shown by Barenboim in \cite{barenboim16sublinear}. The best known algorithm that works for both vertex and edge coloring has a time complexity of $O(\sqrt{\Delta\log\Delta} + \log^* n)$~\cite{fraigniaud16local,BEG18,MausTonoyan20}. For $(2\Delta-1)$-edge coloring, it has recently been shown that we can even obtain a dependency on $\Delta$ that is subpolynomial in $\Delta$. First, Kuhn~\cite{Kuhn20} showed that the problem can be solved in $2^{O(\sqrt{\log\Delta})} + O(\log^* n)$ rounds and subsequently, Balliu, Kuhn, and Olivetti~\cite{BalliuKO20} showed that the number of rounds can even be reduced to $\log^{O(\log\log\Delta)}\Delta + O(\log^* n)$ and thus to a quasi-polylogarithmic dependency on $\Delta$. This leaves a natural open question.

\begin{center}
  \begin{minipage}[h]{0.94\linewidth}
    \em Is it possible to solve $(2\Delta-1)$-edge coloring in time polylogarithmic in $\Delta$?
  \end{minipage}
\end{center}

\paragraph{Our Contribution.} As our main result, we resolve this above open question and prove the following theorem.

\begin{center}
\begin{minipage}{0.99\textwidth}
\begin{mdframed}[hidealllines=false, backgroundcolor=gray!15]
  \begin{theorem}\label{thm:mainsimple}
    There is a deterministic $\poly\log\Delta + O(\log^* n)$-round \LOCAL algorithm to solve the $(2\Delta-1)$-edge coloring problem.
  \end{theorem}
\end{mdframed}
\end{minipage}
\end{center}

In fact, as we will prove \Cref{thm:mainsimple} also for the more general $(\mathit{degree}+1)$-list edge coloring problem, as long as all the colors come from a domain of size at most $\poly(\Delta)$. In this problem, initially, every edge $e$ is given a list consisting of at least $\deg_G(e)+1$ different colors, where the degree $\deg_G(e)$ of $e$ is defined as the number of edges that are adjacent to $e$. The output must be a proper edge coloring such that every edge $e$ is colored with some color from its list. In addition, we also provide a more efficient edge coloring algorithm for the \CONGEST model. The best known \CONGEST algorithm for computing a $(2\Delta-1)$-edge coloring (as a function of $\Delta$) has a round complexity of $O(\Delta+\log^* n)$~\cite{BEG18}. In \cite{BarenboimE11}, it was further shown that for any $\eps\in(0,1]$, a $2^{O(1/\eps)}\cdot\Delta$-edge coloring can be computed in $O(\Delta^{\eps}+\log^* n)$ rounds in the \CONGEST model.\footnote{The authors of \cite{BarenboimE11} do not explicitly show that their algorithm works in the \CONGEST model. It is however not hard to see that it can be adapted to work in the \CONGEST model.} 
We improve this by showing that an $O(\Delta)$-edge coloring can be computed in polylogarithmic (in $\Delta$) time.

\begin{center}
\begin{minipage}{0.99\textwidth}
\begin{mdframed}[hidealllines=false, backgroundcolor=gray!15]
  \begin{theorem}\label{thm:congest}
    For any constant $\eps>0$, there is a deterministic $\poly\log\Delta + O(\log^* n)$-round \CONGEST algorithm to compute an $(8+\eps)$-edge coloring.
  \end{theorem}
\end{mdframed}
\end{minipage}
\end{center}

\paragraph{Further Related Work.} There is also substantial work on randomized algorithms for computing  graph colorings~\cite{PanconesiS97,SchneiderW10symmetry,Barenboim2016,ElkinPS15,HarrisSS18,ChangLP18,HalldorssonKNT22}, MIS~\cite{Barenboim2016,ghaffari16improved}, and maximal matchings~\cite{Barenboim2016}. The best known randomized complexities are $O(\log\Delta + \log^5\log n)$ for MIS~\cite{ghaffari16improved,GGR2020}, $O(\log\Delta + \log^3\log n)$ for maximal matching~\cite{Barenboim2016,fischer17improved}, and $O(\log^3\log n)$ for $(\Delta+1)$-vertex coloring and $(2\Delta-1)$-edge coloring~\cite{ElkinPS15,ChangLP18,GhaffariKuhn20}. For the edge coloring problem, which is the main focus of the present paper, there have also been various results on finding edge colorings with less than $2\Delta-1$ colors. Note first that it is not possible to compute such a coloring in time $O(f(\Delta) + \log^* n)$. In \cite{chang20edgecoloring} (based on techniques developed in \cite{Brandt2016,chang16exponential}), it was shown that for every $\Delta>1$, every deterministic algorithm for computing a $(2\Delta-2)$-edge coloring of $\Delta$-regular trees requires at least $\Omega(\log_\Delta n)$ rounds and every randomized such algorithm requires at least $\Omega(\log_\Delta \log n)$ rounds. When computing an edge coloring with less than $2\Delta-1$ colors, the objective therefore is on the one hand to obtain a coloring that uses not many more than $\Delta$ colors and on the other hand to achieve a time complexity that gets as close as possible to the lower bounds of \cite{chang20edgecoloring}. In recent years, different distributed algorithms that compute edge colorings with $(1+\eps)\Delta$ colors (and even with $\Delta+O(1)$ colors) have been developed~\cite{DubhashiGP98,ElkinPS15,GhaffariKMU18,chang20edgecoloring,SuV19}. For example, by using the randomized algorithm of \cite{chang20edgecoloring} together with derandomization techniques of \cite{FOCS18-derand,Rozhon2020}, one obtains deterministic $\poly\log n$-time and randomized $\poly\log\log n$-time  algorithms for computing edge colorings with $\Delta+\tilde{O}(\sqrt{\Delta})$ colors.

\section{Model and Definitions}
\paragraph{Basic notions}
	Leg $G = (V,E)$ be a graph. We denote with $\Delta$ the maximum degree of $G$, and with $\bar{\Delta}$ the maximum degree of the line graph of $G$, that is, the maximum number of neighboring edges of an edge. Clearly, $\bar{\Delta} \le 2 \Delta - 2$. We denote with $\deg_G(v)$ the degree of a node $v \in V$, and for an edge $e \in E$ we denote with $\deg_{G}(e)$ the degree of edge $e$ in the line graph of $G$, that is, for $e=\set{u,v}$, $\deg_G(e)=\deg_G(u)+\deg_G(v)-2$.  If $G$ is a directed graph, we use $\deg_G(v)$ and $\deg_G(e)$ to denote the degrees of $v$ and $e$ in the undirected version of $G$. If the graph is clear from the context, we may omit it and write $\deg(v)$ and $\deg(e)$. We assume that $\log n$ denotes $\log_2 n$.

\paragraph{List Edge Coloring.}
Assume that for a graph $G = (V,E)$, we are given a set $C = \{1,\ldots,|C|\}$ of colors, called the color space, and for each edge $e = \{u,v\}$, we are given a list of colors $L_e \subseteq C$. The list edge coloring asks to assign a color $c_e \in L_e$ to each edge $e$ such that edges incident to the same node are assigned different colors. In the distributed version of the problem, we assume that all nodes know $C$, both nodes of an edge know the list $L_e$ and at the end, both nodes need to output the color of $e$. The $(\mathit{degree} + 1)$-list edge coloring problem is a special case where $|L_e| \ge \deg_G(e) + 1$ for all $e \in E$. The standard $K$-edge coloring is a special of the list edge coloring problem in which every edge $e$ is given the set $L_e=\set{1,\dots,K}$ as its list.

\paragraph{Relaxed List Edge Coloring.}
In this work, we will make use of a technique that allows us to decompose a hard list coloring instance into many easier ones. This technique has been already used in \cite{barenboim16sublinear,fraigniaud16local,Kuhn20,BalliuKO20}. A list edge coloring instance can be characterized by a parameter $S$, specifying how much larger the lists are as compared to the degree of the edges. A list edge coloring instance is said to have slack at least $S$ if $|L_e| > S \cdot \deg(e)$, for all $e$. We define $P(\bar{\Delta},S,C)$ to be the family of list edge coloring instances where the graph has maximum edge degree $\bar{\Delta}$, the slack is at least $S$, and the color space has size $C$. We define $T(\bar{\Delta},S,C)$ to be the time required to solve $P(\bar{\Delta},S,C)$.
	
\paragraph{Defective Coloring.}
A $d$-defective $c$-coloring is an assignment of colors from
$\{1,\ldots,c\}$ to the nodes, such that the maximum degree of
each graph induced by nodes of the same colors is bounded by $d$.  The
$d$-defective $c$-edge coloring of $G$ is a $d$-defective $c$-coloring
of the line graph of $G$.

\paragraph{\LOCAL and \CONGEST Model.}
We consider two standard models of distributed computing, the \LOCAL
and the \CONGEST model~\cite{linial1987,Peleg2000}.  In the \LOCAL model, the network is
modeled as a graph $G=(V,E)$, where the nodes $V$ represent computational entities and the edges $E$ represent pairwise communication links. Communication proceeds in synchronous rounds, where in each round, each node can send (possibly
different) messages to its neighbors, receive messages from the
neighbors, and perform some internal computation. We do not restrict the
message size or the interal computational power of the nodes.

At the beginning of the computation, each node knows a unique
identifier from $\{1,\ldots,\poly n\}$, where $n = |V|$ is known to
all nodes. Further, each node knows $\Delta$, the maximum degree of
$G$. At the end of a computation, each node must produce its own part
of the solution (e.g., in the case of the edge coloring problem, each
node $v$ must output the colors of all incident edges). The time
complexity of an algorithm is the worst-case of number of rounds
required to produce the solution. We may express this running time as
a function of $n$ and $\Delta$.  The \CONGEST model is similar to the
\LOCAL model, with the only difference that each message must have
size bounded by $O(\log n)$ bits.

\section{Road Map and High Level Ideas}
\label{sec:roadmap}

In this section, we provide an overview of the main ingredients that
we use to solve the edge coloring problem.  While our algorithms do
not work exactly as stated here, we still try to highlight the core
ideas that we use. For this overview, we first sketch how to obtain an edge
coloring with $O(\Delta)$ colors.

At a very high level, our algorithm uses a divide-and-conquer idea
that has been used in various previous deterministic distributed
coloring
algorithms~\cite{Kuhn2009,barenboim14distributed,BarenboimE10,BarenboimE11,Kuhn20,BalliuKO20}. The
algorithm uses a defective coloring to divide the graph into subgraphs
of smaller degree and at the same time the global space of colors is
divided into the same number of parts so that the different low-degree
subgraphs can use disjoint color spaces. The colorings or the
different subgraphs are then colored recursively in parallel. There
are different challenges that we face with this high-level
approach. First, when computing a defective coloring with defect $d$,
all previous algorithms use a number of colors that is at best
$c\cdot \Delta / d$ for some constant $c>1$. Because of this, the ratio between
the necessary number of colors and the maximum degree grows by a
factor $c$ for every recursion level. To keep the total number of
colors needed in the end moderately small, this requires to keep the
number of recursion levels small, which in return makes the defective
coloring steps more expensive. As the main technical contribution of
this paper, we obtain a new algorithm for defective edge coloring,
where we can keep the number of colors and the reduction of the degree
arbitrarily close to $1$.

Concretely, for the special case of $2$-colored bipartite graphs, we
obtain a defective $2$-coloring algorithm with the guarantee that the
defect of every edge $e$ is only $(1/2+\eps)\cdot \deg(e)$ for an
arbitrary parameter $\eps>0$ (as long as the defect of the edge does
not fall below some threshold). We can compute this defective
$2$-coloring in time $\poly(\log(\Delta)/\eps)$. By choosing
$\eps \leq \eps'/\log\Delta$ and recursively applying the defective
$2$-coloring, we show that we can partition the graph into
$2^k=\Delta/\poly\log\Delta$ parts such that the defect of each node
$v$ is at most $(1+O(\eps'))\cdot\deg_G(e)/2^k$. We can thus compute a
$(2+\eps')\Delta$-edge coloring of a given $2$-colored bipartite graph
in time $\poly(\log(\Delta)/\eps')$. We further give a reduction that
allows to use this algorithm to color general graphs with
$(8+\eps)\Delta$ colors in time $\poly\log\Delta + O(\log^* n)$.

We next describe the high-level idea of our defective $2$-coloring algorithm. Assume that we are given a bipartite graph $G=(U\cup V, E)$, where the nodes know if they are in $U$ or in $V$. For simplicity, assume further that $G$ is $\Delta$-regular. Our goal is to color each edge either red or blue such that every red edge has at most $(1+\eps)\Delta$ adjacent red edges and every blue edge has at most $(1+\eps)\Delta$ adjacent blue edges. To achieve this, we generalize an idea that was presented in \cite{tokendropping}. There, the authors show how to efficiently solve a problem known as locally optimal semi-matching~\cite{CzygrinowHSW16} and in particular a special case of this problem, a so-called stable edge orientation of a graph. Given an edge orientation of a graph $G$, for every node $v$, let $x_v$ denote the number of incident edges that are oriented towards $v$. An edge orientation is called stable if for every edge $e=\set{u,v}$, if $e$ is oriented from $u$ to $v$, then $x_v-x_u\leq 1$ (and otherwise $x_u-x_v\leq 1$). Note that such a stable orientation of a $\Delta$-regular bipartite graph $G=(U\cup V, E)$ directly gives a perfect defective $2$-coloring. Assume that every edge that is oriented from $U$ to $V$ is colored red and every edge that is oriented from $V$ to $U$ is colored blue. If an edge $\set{u,v}$ for $u\in U$ and $v\in V$ is oriented from $u$ to $v$ and thus colored red, then the number of adjacent red edges is exactly $(x_v-1) + (\Delta-x_u-1) = \Delta +(x_v-x_u)-2\leq \Delta-1$ (by using that $x_v-x_u\leq 1$). Note that the degree of every edge in a $\Delta$-regular graph is exactly $2(\Delta-1)$. An analogous argument works for blue edges.

In \cite{tokendropping}, the authors introduce a tool that they call the \emph{token dropping game} and which is used in particular to compute stable edge orientations. The token dropping game works as follows. We are given a directed graph where each node holds either $0$ or $1$ tokens. The goal is to move tokens around and reach a ``stable'' solution, that is a solution in which tokens cannot move. A token can move from node $u$ to node $v$ if $u$ has a token, $v$ has no tokens, and the edge $(u,v)$ exists. Every time a token passes through an edge, the edge is deleted. In \cite{tokendropping}, it is shown that the token dropping game can be solved in time $O(L\cdot\Delta^2)$ if the directed graph of the game has no cycles and the longest directed path is of length $L$. It is shown that this algorithm can be used to compute a stable edge orientation and thus a perfect defective $2$-coloring in $\Delta$-regular $2$-colored bipartite graphs in $O(\Delta^4)$ rounds. 

We cannot directly apply the algorithm of \cite{tokendropping} for several reasons. First, a time complexity of $O(\Delta^4)$ is way too slow for us as we aim for algorithms that run in time polylogarithmic in $\Delta$. Second, we need to deal with non-regular graphs. Even if at the beginning the graph is regular, after a single application of defective $2$-coloring, we might end up with two very non-regular graphs.\footnote{Note that even if the new edge degrees are at most $\Delta-1$, the node degrees can still take arbitrary values between $0$ and $\Delta$.} In non-regular graphs, a stable orientation does not lead to a good defective $2$-coloring. The second problem can be solved relatively easily. For each edge $e=\set{u,v}$, one can just define a different threshold $\eta_e$ such that if $e$ is oriented from $u$ to $v$, it must hold that $x_v-x_u\leq\eta_e + 1$ and if $e$ is oriented from $u$ to $v$, it must hold that $x_u-x_v\leq -\eta_e+1$. If the thresholds $\eta_e$ are chosen in the right way, we can still get a perfect defective $2$-coloring from such an edge orientation and one can also still use the reduction in \cite{tokendropping} to the token dropping game to compute such an edge orientation. Reducing the time complexity is more challenging. For this, we relax the requirement on the orientation. For $\Delta$-regular graphs, we now want to require that if an edge $\set{u,v}$ is oriented from $u$ to $v$, then $x_v-x_u\leq \eps\Delta$ and otherwise $x_u-x_v\leq \eps\Delta$. Such an orientation can be computed fast if we have a fast algorithm to a generalized token dropping game, where each node can have up to $O(\eps\Delta)$ tokens. For more details of how to define the generalized token dropping game, we refer to \Cref{sec:tokendropping}.

As long as the directed graph of the token dropping game has no cycles and only short directed paths, it is in fact possible to adapt the token dropping algorithm and the analysis of \cite{tokendropping} (in a non-trivial way) to obtain a $\poly(\log(\Delta)/\eps)$-round algorithm to compute a defective $2$-coloring in a $2$-colored bipartite graph $G$ such that the defect of any edge $e$ of $G$ is at most $(1+\eps)\deg_G(e)$ and this is sufficient to obtain an $(8+\eps)\Delta$-edge coloring algorithm with a time complexity of $\poly\log\Delta+O(\log^* n)$. In order to obtain a $(2\Delta-1)$-coloring, there is however still one major challenge remaining. When recursively using defective colorings, we always have to slightly relax the coloring problem along the way and end up with more that $2\Delta-1$ colors. We can use such a more relaxed edge coloring to compute a $(2\Delta-1)$-edge coloring or even a $(\mathit{degree}+1)$-edge coloring if we solve the more general list edge coloring~\cite{barenboim16sublinear,fraigniaud16local}. In this case, one however also has to use a generalized version of defective coloring~\cite{Kuhn20,BalliuKO20} (for a definition of what we need, see \Cref{sec:gen-def-2-ec}). Even for this generalized defective $2$-coloring variant, one can define appropriate conditions for the required edge orientation and one can use our generalized token dropping game to compute such an edge orientation. However, in this case, the resulting token dropping graph can have cycles and therefore, even the adapted variant of the token dropping algorithm of \cite{tokendropping} does not work. In \Cref{sec:tokendropping}, we therefore design a completely new token dropping algorithm, which also works in general directed graphs (for the relaxed token dropping game that we are using).

The remainder of the paper is organized as follows. In \Cref{sec:tokendropping}, we introduce and solve the generalized token dropping game. In \Cref{sec:gen-def-2-ec}, we then introduce our generalized defective $2$-coloring problem and we show how to solve it by using our token dropping algorithm of \Cref{sec:tokendropping}. Finally, in \Cref{sec:econgest} and \Cref{sec:eclocal}, we prove our main results. We start in \Cref{sec:econgest} with the $O(\Delta)$-edge coloring algorithm for the \CONGEST model, which is conceptually simpler. Then, in \Cref{sec:eclocal}, we present our algorithm for solving the $(\deg(e)+1)$-list edge coloring problem in the \LOCAL model.

\section{The Generalized Token Dropping Game}
\label{sec:tokendropping}

A key technical tool that we use is a generalization of the token dropping game of \cite{tokendropping}. This game is defined on a directed graph $G=(V,E)$, and it has an integer parameter $k\geq 1$. Initially, each node $v\in V$ receives at most $k$ tokens as input. Each edge $e\in E$ can be either active or passive. Initially, all edges are active, and as the game progresses, edges can become passive. In a sequential execution, the game proceeds in steps, where in each step, on some active edge $(u,v)\in E$ on which $u$ has at least $1$ token and $v$ has less than $k$ tokens, one token can be moved from $u$ to $v$. After moving the token, the edge $(u,v)$ becomes passive. Passive edges cannot become active again. That is, over each edge of $G$, at most one token can be moved, and an edge is passive if and only if a token was moved over the edge. Let $\tau(v)$ be the number of tokens at a node $v$ at a given time. When the game ends, it must hold that $\tau(v)\leq k$ for every $v\in V$ and 
\begin{equation}\label{eq:tokenedgecondition}
  \forall e=(u,v)\in E\,:\, e\text{ active } \Longrightarrow \tau(u) \leq \tau(v) + \sigma(e),
\end{equation}
where $\sigma(e)\geq 0$ is a value that specifies how much slack we tolerate on a given edge $e$. In the original token dropping game introduced in \cite{tokendropping}, $k=1$ and $\sigma(e)=0$ for all edges. In \cite{tokendropping}, it was further assumed that the graph $G$ is organized in layers and that all edges are oriented from higher to lower layers (therefore the name token dropping). We here generalize the game by allowing general diriected graphs, larger values of $k$, and by tolerating some slack on the active edges.

\subsection{Distributed Token Dropping Algorithms}

In the distributed version of the (generalized) token dropping game, the goal is to run an execution of the game, where tokens are moved in parallel, but which is still equivalent to the sequential definition of the game. That is, over each edge, at most one token can be moved and at all times, and every node has a set of at most $k$ tokens.

We next describe a distributed algorithm to solve the generalized token dropping game with $\sigma(e)=\eps\cdot\deg_G(e)$ for some given parameter $\eps\in(0,1/2]$, that is, by allowing a slack proportional to the edge degree. The algorithm has an integer parameter $\delta>0$, and we will see that $\delta$ can be used to control the trade-off between the round complexity and the slack of the algorithm. The smaller $\delta$ is chosen, the smaller $\eps$ can be chosen, however the time of the algorithm depends linearly on $1/\delta$.

Throughout the algorithm, some tokens are active and some tokens are passive. Initially, all tokens are active, and once a token becomes passive, it remains passive and cannot be moved anymore. The algorithm operates in synchronous phases. We use $x_v(t)$ and $y_v(t)$ to denote the number of active and passive tokens of node $v$ at time $t$, i.e., at the end of phase $t$. The value $x_v(0)$ and $y_v(0)$ denote the number of tokens of $v$ at the beginning. At all times $t\geq 0$, the algorithm always guarantees that $x_v(t)+y_v(t)\leq k$ for every node $v\in V$. We further have $y_v(0)=0$ for all $v\in V$. For each node, we further define an integer parameter $\alpha_v\geq 1$ that controls how much slack node $v$ is willing to tolerate on its edges. We will see that if we want to tolerate slack $\sigma_e$ on an edge $\set{u,v}$, then we in particular have to choose $\alpha_u,\alpha_v\leq c\sigma_e$ for a sufficiently small constant $c$.
The value of $\alpha_v$ controls how much slack node $v$ is willing to tolerate on its edges. In the end, the slack on each edge $(u,v)$ that is still active has to be $O(\alpha_u+\alpha_v)$. The algorithm is run for $\lfloor\frac{k}{\delta}\rfloor-1$ phases. In each phase $t\geq 1$, the algorithm does the following steps.

\begin{enumerate}
\item Define $A(t)\subseteq V$ as the set of nodes $v\in V$ with $x_v(t-1)\geq \alpha_v + \delta$. We call $A(t)$ the active nodes in phase $t$ and only nodes in $A(t)$ will be able to move tokens to other nodes in phase $t$.
\item Each node $v\in A(t)$ sets $x_v'(t) := x_v(t-1) - \delta$ and $y_v'(t) := y_v(t-1)+\delta$. All other nodes $v\in V\setminus A(t)$ set $x_v'(t):=x_v(t-1)$ and $y_v'(t) := y_v(t-1)$.
\item For a node $v\in V$, let $S(v)\subset A(t)$ be the set of nodes $u\in A(t)$ such that there is an active edge $(u,v)$ from $u$ to $v$ ($S(v)$ are the nodes that can potentially send a token to $v$ in phase $t$).
\item If $x_v'(t)\leq k - t\delta - \alpha_v$, $v$ sends a token proposal to the $\min\set{|S(v)|,k-t\delta-x_v'(t)}$ nodes $w\in S(v)$, where priority is given to nodes $w\in S(v)$ of smaller $\deg_G(w)/\alpha_w$ value.
\item For each node $v\in V$, let $p_v(t)$ be the number of token proposals $v$ receives in phase $t$ and let $q_v(t):=\min\set{p_v(t),x_v'(t)}$. Node $v$ accepts $q_v(t)$ of the proposals and sends a token to the respective outneighbors. The edges over which a token is sent become passive.
\item For each node $v\in V$, let $r_v(t)$ be the number of tokens that $v$ receives in phase $t$. The number of active tokens at the end of phase $t$ of each node $v$ is set to $x_v(t):=x_v'(t)+r_v(t)-q_v(t)$, i.e., $x_v(t)$ is $x_v'(t)$ plus the number of received tokens and minus the number of sent tokens. 
\end{enumerate}

We start by proving that the maximum number of active tokens decreases after each phase, and that the total number of tokens at each node never exceeds $k$.
\begin{lemma}\label{lemma:tokenbounds}
  For all $v\in V$ and for all $t\geq 0$, if $\delta\leq \alpha_v$, it holds that $x_v(t)\leq \max\set{2\alpha_v, k - t\cdot \delta}$ and $y_v(t)\leq k - x_v(t)$.
\end{lemma}
\begin{proof}
  We prove the upper bound on $x_v(t)$ by induction on $t$. For $t=0$,
  we have $x_v(0)\leq k$ and therefore clearly
  $x_v(0)\leq \max\set{\alpha_v, k - 0\cdot \delta}$. For the
  induction step, let us assume that $t\geq 1$ and let us therefore
  focus on what happens in phase $t$. Note that in step 4 of phase $t$
  the above algorithm, each node $v$ sends at most
  $k-t\delta - x_v'(t)$ proposals and it therefore receives at most
  that many new tokens. If $x_v'(t)\leq k - t\delta$, we therefore
  also have $x_v(t)\leq x_v'(t) + (k-t\delta-x_v'(t))=k-t\delta$. In
  this case, we have proven the required upper bound on $x_v(t)$. Let
  us therefore assume that $x_v'(t)>k-t\delta$. In this case, $v$ does
  not send any proposals and we therefore know that
  $x_v(t)\leq x_v'(t)$. By the induction hypothesis, we also know that
  $x_v(t-1)\leq \max\set{2\alpha_v, k - (t-1)\cdot \delta}$. Since we
  know that $x_v'(t)\leq x_v(t-1)$, then, whenever $x_v(t-1)\leq 2\alpha_v$,
  we now also have $x_v(t)\leq x_v'(t)\leq 2\alpha_v$. Let us
  therefore further assume that
  $2\alpha_v< x_v(t-1)\leq k-(t-1)\delta$. Because
  $\delta\leq\alpha_v$, in this case, we have
  $x_v(t-1)> 2\alpha_v \geq \alpha_v + \delta$ and therefore
  $v\in A(t)$. However, we then set
  $x_v'(t) := x_v(t-1)-\delta \leq k- t\delta$, and we therefore also
  have $x_v(t)\leq k-t\delta$. The proves the upper bound on $x_v(t)$.

  It remains to show that $y_v(t)\leq k-x_v(t)$. We again apply
  induction on $t$. Note that we have $y_v(0)=0$ and $x_v(0)\leq k$
  and therefore the bound certainly holds for $t=0$. For every
  $t\geq 1$, we either have $y_v(t)=y_v(t-1)$ and $v\not\in A(t)$ or
  $y_v(t)=y_v(t-1)+\delta$ and $v\in A(t)$. We
  therefore in particular always have $y_v(t)\leq y_v(t-1)+\delta$ and thus
  $y_v(t)\leq t\delta$. If $x_v(t)\leq k-t\delta$, we thus have
  $x_v(t)+y_v(t)\leq k$ as required. Let us therefore assume that $x_v(t)>k-t\delta$. Above, we showed that if
  $x_v'(t)\leq k-t\delta$, it follows that $x_v(t)\leq k-t\delta$. By the contrapositive, if $x_v(t)>k-t\delta$, we therefore know that also $x_v'(t)>k-t\delta$ and thus that $v$ does not sent any proposals in step 4 of the algorithm. In this case, we therefore know that $x_v(t)\leq x_v'(t)$. If $v\not\in A(t)$,
  we have $y_v(t)=y_v(t-1)$ and $x_v'(t)=x_v(t-1)$ and if $v\in A(t)$,
  we have $y_v(t)=y_v(t-1)+\delta$ and $x_v'(t)=x_v(t-1)-\delta$. In
  both cases, $x_v(t)\leq x_v'(t)$ and the induction hypothesis
  $x_v(t-1)+y_v(t-1)\leq k$ directly imply that $x_v(t)+y_v(t)\leq k$.
\end{proof}

We now prove that, for each active edge, the number of passive tokens of each endpoint cannot differ by much.
\begin{lemma}\label{lemma:passiveslackbound}
  For every edge $e=(u,v)\in E$, at the end of each phase $t\leq  k/\delta  -1$ in the above algorithm, $e$ is passive or we have
  \[
    y_u(t)- y_v(t) \leq 2\alpha_v + \left(\frac{\deg_G(u)\cdot\deg_G(v)}{\alpha_u\cdot\alpha_v}+\frac{\deg_G(u)}{\alpha_u}+\frac{\deg_G(v)}{\alpha_v}\right)\cdot \delta.
  \]
\end{lemma}
\begin{proof}
  In each phase $t$, for every node $w\in V$, we have
  $y_w(t)=y_w(t-1)+\delta$ if $w\in A(t)$ and $y_w(t)=y_w(t-1)$
  otherwise. To upper bound $y_u(t)-y_v(t)$ while $e=(u,v)$ is active,
  we can therefore count the number of phases in which $u\in A(t)$ and
  $v\not\in A(t)$ and where no token is passed over the edge. For $u$
  to be in $A(t)$, we must have $x_u(t-1)\geq \alpha_u + \delta$ and
  for $v$ to not be in $A(t)$, we must have
  $x_v(t-1)<\alpha_v+\delta$.

  As long as $t$ is not too large, $x_v(t-1)<\alpha_v+\delta$ implies
  that $v$ has the capacity to receive tokens in phase $t$ and it can
  therefore send proposals to active in-neighbors.  To count the
  number of phases in which $u\in A(t)$ and $v\not\in A(t)$, we
  therefore make a case distinction depending on the value of $t$. We
  first assume that $t\leq (k-2\alpha_v)/\delta - 1$. We can rewrite
  this as $\alpha_v+\delta \leq k-t\delta -\alpha_v$ and because we
  know that $x_v(t-1)\leq\alpha_v+\delta$, this implies that
  $x_v(t-1)\leq k - t\delta - \alpha_v$. Because $v\not\in A(t)$, we
  also know that $x_v'(t)=x_v(t-1)$ and therefore
  $x_v'(t)\leq k-t\delta-\alpha_v$. This is the condition that is
  needed in step 4 of the above algorithm for $v$ to send token
  proposals over its incoming active edges. Whenever $v$ sends a token
  proposal to an in-neighbor $w$, either $w$ sends a token to $v$ and
  afterwards the edge $(w,v)$ becomes passive, or $w$ sends at least
  $\alpha_w$ tokens to other out-neighbors and therefore at least
  $\alpha_w$ out-edges of $w$ become passive. To each in-neighbor
  $w\neq u$, $v$ can therefore send proposals at most
  $\lceil \deg_G(w)/\alpha_w\rceil$ times. Note that in each phase $t$,
  $v$ only sends proposals (in step 4) if
  $x_v'(t)\leq k - t\delta - \alpha_v$ and it sends proposals either
  to all nodes in $S(v)$ (i.e., to all active neighbors with an active
  edge to $v$) or it sends proposals to
  $k-t\delta-x_v'(t)\geq \alpha_v$ nodes in $S(v)$. Hence, in each
  phase $t$ in which $v$ sends proposals and in which $v$ does not
  send a proposal to $u$ although $u\in A(t)$ and $(u,v)$ is still
  active, there must be at least $\alpha_v$ neighbors
  $w\in S(v)\setminus\{u\}$ to which $v$ sends a proposal. Note that
  for each such neighbor $w$, we have
  $\deg_G(w)/\alpha_w\leq \deg_G(u)/\alpha_u$ (because in step 4,
  proposals are sent to active in-neighbors of smallest
  $\deg_G(w)/\alpha_w$ ratio first). The number of such phases $t$ for
  $t\leq (k-2\alpha_v)/\delta - 1$ in which $u\in A(t)$ and
  $v\not\in A(t)$, in which $(u,v)$ is active and $v$ does not send
  a proposal to $u$ can therefore be upper bounded by
  \begin{eqnarray}
    \left\lfloor\frac{(\deg_G(v)-1)\cdot\left\lceil\frac{\deg_G(u)}{\alpha_u}\right\rceil}{\alpha_v}\right\rfloor
    & \le & \left(\frac{\deg_G(v)-1}{\alpha_v}\right)\cdot \left(\frac{\deg_G(u)}{\alpha_u}+1\right)\nonumber\\
    & \le & \frac{\deg_G(u)\cdot \deg_G(v)}{\alpha_u\cdot\alpha_v} + \frac{\deg_G(v)}{\alpha_v}. \label{eq:nofincreasephases}
  \end{eqnarray}
  Further, $v$ can send at most $\lfloor \deg_G(u)/\alpha_u\rfloor$ proposals to $u$ without receiving a token from $u$ (and thus such that $(u,v)$ remains active).
  The total number of phases $t$ for $t\leq k/\delta-1$ in which $u\in A(t)$ and
  $v\not\in A(t)$ and in which $(u,v)$ remains active can consequently be
  upper bounded by
  \[
    \frac{2\alpha_v}{\delta} + \frac{\deg_G(u)\cdot \deg_G(v)}{\alpha_u\cdot\alpha_v} + \frac{\deg_G(v)}{\alpha_v} + \frac{\deg_G(u)}{\alpha_u}.
  \]
  In each of those phases, $y_u(t)-y_v(t)$ increases by $\delta$, which directly implies the claim of the lemma.
\end{proof}

We are now ready to prove that our algorithm solves the generalized token dropping game. The following theorem follows almost directly from \Cref{lemma:tokenbounds}, for space reasons, the proof appears in \Cref{app:deferredtoken}.

\begin{theorem}\label{thm:tokenbound}
  At the end of the above algorithm, for every $v\in V$, let $\tau(v)$ be the number of tokens at node $v$.
  If for all $v\in V$, $\alpha_v\geq \delta$, the above algorithm has a time complexity of $O(k/\delta)$ and at the end of the algorithm, for every node $v\in V$, we have $\tau(v)\leq k$ and for every edge $(u,v)$, either $(u,v)$ is passive or
  \[
    \tau(u)-\tau(v) \leq 2(\alpha_u+\alpha_v) + \left(\frac{\deg_G(u)\cdot \deg_G(v)}{\alpha_u\cdot\alpha_v} + \frac{\deg_G(u)}{\alpha_u} +
      \frac{\deg_G(v)}{\alpha_v}\right)\cdot\delta.
  \]
\end{theorem}


\section{Generalized Defective 2-Edge Coloring}\label{sec:gen-def-2-ec}
At the core, our edge coloring algorithms are based on the following simple idea. We partition the space of possible colors into two parts and each edge commits to choosing a color from one of the two parts. The two parts can then be solved recursively. Because edges that pick colors from disjoint color spaces cannot conflict with each other, the two parts can be colored recursively in parallel. The task of splitting the set of edges into two parts can be formulated as a defective edge coloring problem as follows.

\begin{definition}[Generalized Defective $2$-Edge Coloring]\label{def:defective2coloring}
  Given values $\eps\geq 0$ and $\beta\geq 0$, a graph $G=(V,E)$, and parameters $\lambda_e\in[0,1]$ for all edges $e\in E$, a generalized $(1+\eps,\beta)$-relaxed defective $2$-edge coloring of $G$ is an assignment of colors red and blue to the edges $e\in E$ such that for every edge $e\in E$:
  \begin{itemize}
  \item If $e$ is colored red, the number of neighboring red edges is $\leq(1+\eps)\cdot\lambda_e\cdot \deg_G(e) + \lambda_e\beta$.
  \item If $e$ is colored blue, the number of neighboring blue edges is $\leq (1+\eps)\cdot(1-\lambda_e)\cdot \deg_G(e) + (1-\lambda_e)\beta$.
  \end{itemize}
\end{definition}

We will next show how we can solve a given generalized defective $2$-edge coloring instance in two-colored bipartite graphs by using the token dropping game of \Cref{sec:tokendropping}. We next show how to transform the generalized defective $2$-coloring problem to make it more directly amenable to applying the token dropping game. We first define the notion of generalized balanced edge orientations. For convenience, we only give a definition for bipartite graphs.

\begin{definition}[Generalized Balanced Edge Orientation]\label{def:balancedorientation}
  Assume that we are given values $\eps\geq 0$ and $\beta\geq 0$, a bipartite graph $G=(U\dot{\cup}V,E)$, parameters $\eta_e\in\mathbb{R}$ for all edges $e=(u,v)\in E$, $u\in U$, $v\in V$, and an orientation on the edges of $G$. For each node $w\in U\cup V$, let $x_w$ be the number of edges of $w$ that are oriented towards $w$. The orientation is called a \emph{generalized $(\eps,\beta)$-balanced edge orientation} of $G$ if the following holds.  For every edge $(u,v)\in E$,
  \begin{itemize}
  \item[(I)] If $e$ is oriented from $u$ to $v$, then $x_v-x_u\leq \eta_e + 1 + \frac{\eps}{2}\cdot\deg_G(e)+\beta$.
  \item[(II)] If $e$ is oriented from $v$ to $u$, then $x_u-x_v\leq -\eta_e + 1 + \frac{\eps}{2}\cdot\deg_G(e)+\beta$.
  \end{itemize}
\end{definition}

The following lemma follows almost directly from the above definitions. The proof appears in \Cref{app:deferreddefective}.
\begin{lemma}\label{lemma:orientation2defective}
  Assume that we are given a bipartite graph $G=(U\dot{\cup} V,E)$, a
  parameter $\eps\geq 0$, and parameters $\lambda_e\in[0,1]$ for all
  $e\in E$. For every edge $e=(u,v)$, $u\in U$, $v\in V$, we define
  \begin{equation}
    \label{eq:def_eta}
    \eta_e := 1 -2\lambda_e - (1-\lambda_e)\cdot \deg_G(u) +\lambda_e\cdot\deg_G(v) + \eps\cdot\left(\lambda_e- \frac{1}{2}\right)\cdot\deg_G(e) + (2\lambda_e -1)\beta
  \end{equation}
  A generalized $(\eps,\beta)$-balanced edge orientation of $G$ w.r.t.\ the edge parameters $\eta_e$ can be turned into a solution to the given generalized $(1+\eps,2\beta)$-relaxed defective $2$-edge coloring w.r.t.\ the original edge parameters $\lambda_e$ by coloring edges red that are oriented from $U$ to $V$ and by coloring edges blue that are oriented from $V$ to $U$.
\end{lemma}

We next show how to compute a generalized balanced edge orientation (as given by \Cref{def:balancedorientation}) and thus in combination with \Cref{lemma:orientation2defective} a generalized defective $2$-edge coloring. Assume that we have a bipartite $2$-colored graph $G=(U\cup V,E)$ with edge parameters $\eta_e\in\mathbb{R}$. We will compute a generalized balanced edge orientation with parameter $\eta_e$ by reducing it to a sequence of instances of the token dropping game. More concreteley, the algorithm has a parameter $\nu>0$ and it runs in a sequence of phases $\phi=1,2,3,\dots,O\big(\frac{\log{\Delta}}{\nu}\big)$. At the start, all edges of $G$ are unoriented and in each phase, a subset of the unoriented edges become oriented. We define $F_{\phi}$ as the set of edges that get oriented in phase $\phi$, and $F_{<\phi}$ as the set of edges that get oriented before phase $\phi$. By writing $\deg_{F_{<\phi}}(v)$ we refer to the degree of node $v$ in the graph induced by the edges $F_{<\phi}$.
 In each phase, we use one instance of the token dropping game to make sure that the set of all oriented edges satisfies inequalities (I) and (II) of \Cref{def:balancedorientation} (for an appropriate value of $\eps$). For an edge $e\in E$ and a phase $\phi\geq 0$, we use $d(e,\phi)$ to denote the number of unoriented neighboring edges of $e$ at the end of phase $\phi$. For convenience, we also use $d(e,0)=\deg_G(e)$ to denote the number of unoriented neighboring edges of $e$ at the start. We further define $\bar{\Delta}:=2\Delta-2$ as an upper bound on the maximum edge degree in $G$. We further set the parameter $\nu$ such that
\begin{equation}\label{eq:def_nu}
  0 < \nu \leq \frac{1}{8}.
\end{equation}
For every node $v$ and every phase $\phi\geq 1$, let $x_v(\phi)$ denote the number of edges that are oriented towards $v$ at the end of phase $\phi$. For convenience, we also define $x_v(0)=0$ as at the beginning all edges are unoriented and therefore no edges are oriented towards $v$. For each node $v$, we further define
\begin{equation}\label{eq:degmin_alpha}
  \degmin(v) := \min_{e\in F_{<\phi} : v\in e} \deg_G(e)\qquad\text{and}\qquad
  \phalpha{v} := \max\set{1,\frac{1}{4}\cdot \frac{\nu^2}{\ln{\bar{\Delta}}}\cdot (\degmin(v)+1)}.
\end{equation}
We assume that the nodes of $G$ know if they are in $U$ or in $V$ (i.e., we assume that the bipartite graph $G$ is equipped with a $2$-vertex coloring). The algorithm in phase $\phi\geq 1$ works as follows:
\begin{enumerate}
\item Let $E_{\phi}\subseteq E$ be the edges $e\in E$ that are still unoriented at the beginning of phase $\phi$ and for which $d(e,\phi-1)>(1-\nu)^\phi\bar{\Delta}$.
\item For every edge $e=(u,v)\in E_{\phi}$ with $u\in U$ and $v\in V$,  $e$ sends a proposal to $v$ if $x_v-x_u\leq \eta_e$ and it sends a proposal to $u$ otherwise.
\item We set $k_{\phi} := \big\lceil \nu(1-\nu)^{\phi-1}\bar{\Delta}\big\rceil$. Every node $u\in U\cup V$ that receives at most $k_{\phi}$ proposals from its edges in $E_{\phi}$ accepts all those proposals and every node $u\in U\cup V$ that receives more that $k_{\phi}$ proposals from its edges in $E_{\phi}$ accepts an arbitrary subset of $k_{\phi}$ of them. 
\item Let $F_{\phi}\subseteq E_{\phi}$ be the set of edges for which the proposal gets accepted. The edges $e\in F_{\phi}$ will be the ones that get newly oriented in phase $\phi$. For each edge $(u,v)\in F_{\phi}$, the edge is oriented from $u$ to $v$ if $(u,v)$'s proposal was sent to and accepted by $v$ and the edge is oriented from $v$ to $u$ otherwise.
\item Let $F_{<\phi}:=\bigcup_{i=1}^{\phi-1} F_{\phi}$ be the set of edges that get oriented before phase $\phi$. We define a subset $F'_{<\phi}\subseteq F_{<\phi}$ of those edges as follows. An edge $e=(u,v)\in F_{<\phi}$ ($u\in U$, $v\in V$) is included in $F'_{<\phi}$ if either $e$ is oriented from $u$ to $v$ and $x_v(\phi-1)-x_u(\phi-1)>\eta_e $ or if $e$ is oriented from $v$ to $u$ and $x_u(\phi-1)-x_v(\phi-1)>-\eta_e $.
\item We now run an instance of the token dropping game on the graph $G_{\phi}=(U\cup V, F'_{< \phi})$ (i.e., on the subgraph of $G$ induced by the edges in $F'_{< \phi}$), where each edge in $F'_{< \phi}$ is directed in the opposite direction of its current orientation. Each node $u\in U\cup V$ uses the parameter $\phalpha{u}$ as fixed in \Cref{eq:degmin_alpha}. Further, the initial number of tokens of each node $u\in U\cup V$ is equal to the number of proposals from its edges in $E_{\phi}$, $u$ has accepted in the above step 3. Finally, the parameter $\delta_{\phi}$ is set to
  \begin{equation}\label{eq:delta_def}
    \delta_{\phi} := \max\set{1, \left\lfloor \frac{1}{16}\cdot \frac{\nu^6}{\ln^3\bar{\Delta}}\cdot (1-\nu)^{\phi-1}\bar{\Delta}\right\rfloor}.
  \end{equation}
\item To conclude phase $\phi$, we now update the orientation of the edges in $F_{< \phi}$ as follows. We switch the orientation of each edge over which a token is moved in the above token dropping game instance of step 6. All other edges in $F_{< \phi}$ keep their orientations.
\end{enumerate}

We first show that the maximum edge degree of the unoriented part of $G$ decreases exponentially as a function of the number of phases. The proof of \Cref{lemma:degreereduction} is simple and appears in \Cref{app:deferreddefective}.

\begin{lemma}\label{lemma:degreereduction}
  At the end of phase $\phi\geq 1$ of the above algorithm, we have $d(e,\phi)\leq (1-\nu)^{\phi}\bar{\Delta}$ for every edge $e\in E\setminus F_{\leq\phi}$, that is, for every edge that is still unoriented after phase $\phi$.
\end{lemma}

To analyze  the quality of the produced edge orientation, we define the following quantities for every edge $e\in E$.
\begin{equation}
  \label{eq:k_xi_edge}
  k_e := \left\lceil\frac{\nu}{1-\nu}\cdot\deg_G(e)\right\rceil\qquad\text{and}\qquad
  \xi_e := \frac{5}{2}\cdot \frac{\nu}{\ln\bar{\Delta}}\cdot k_e +
             28\cdot \frac{\ln^2\bar{\Delta}}{\nu^4}.
\end{equation}

\begin{lemma}\label{lemma:orientationgrowth}
  After $\phi$ phases of the above algorithm, for every edge $e=(u,v)\in E$ ($u\in U$, $v\in V$), it holds that either: 
  \begin{itemize}
  \item $e$ is unoriented, or
  \item $e$ is oriented from $u$ to $v$ and $x_v(\phi)-x_u(\phi)\leq \eta_e + k_e + \phi\cdot\xi_e$, or
  \item $e$ is oriented from $v$ to $u$ and $x_u(\phi)-x_v(\phi)\leq -\eta_e + k_e + \phi\cdot\xi_e$.
  \end{itemize}
\end{lemma}
\begin{proof}
  We prove the lemma by induction on the number of phase $\phi$. At
  the beginning, all edges are unoriented and therefore clearly for
  $\phi=0$, the claim of the lemma holds. For the induction step,
  assume that $\phi\geq 1$.  Consider some edge $(u,v)\in E$ with
  $u\in U$ and $v\in V$ and assume that $e$ is oriented at the end of
  phase $\phi$ (as otherwise, the claim of the lemma clearly holds for
  $e$). We first show that for every node $v\in U\cup V$, we have
  \begin{equation}
    \label{eq:xv_bounds}
    x_v(\phi-1)\leq x_v(\phi) \leq x_v(\phi-1) + k_{\phi}.
  \end{equation}
  To see why \Cref{eq:xv_bounds} holds, first consider the number of
  incoming edges at node $v$ in the middle of phase $\phi$, right
  before the we run the token dropping instance of phase $\phi$ in
  step 6 of the above algorithm. As the orientation of the edges that
  have been oriented prior to phase $\phi$ can only be changed during
  token dropping game, the number of incoming edges of $v$ at this
  point is equal to $x_v(\phi-1)$ plus the number of newly oriented
  edges of $v$ that are oriented towards $v$. Let us use $y_v$ to
  denote this number of those newly oriented edges towards $v$. Note
  that $y_v$ is exactly equal to the number of proposals $v$ accepts
  in step 3 of the above algorithm. Note also that the initial number
  of tokens of node $v$ in the token dropping game of phase $\phi$ is
  equal to $y_v$. When running token dropping, whenever moving a token
  from a node $u$ to a node $u'$, an edge that was previously oriented
  from $u'$ to $u$ is reoriented from $u$ to $u'$. Hence, for every
  token that node $v$ receives in the token dropping game, the total
  number of edges oriented to $v$ increases by $1$ and for every token
  that $v$ moves to a neighbor, the total number of edges oriented to
  $v$ decreases by $1$. By induction on the steps of the token
  dropping game, the total number of incoming edges at node $v$ at the
  end of phase $\phi$ is therefore exactly $x_v(\phi-1)$ plus the
  number of tokens at node $v$ at the end of the token dropping game
  instance of phase $\phi$. Because by definition of the token
  dropping game, the number of tokens at each node is always in
  $\set{0,\dots,k_{\phi}}$, \Cref{eq:xv_bounds} follows.

  Recall that if an edge $e$ gets oriented in a phase $\phi'$, we have $\deg_G(e)\geq d(e,\phi'-1)\geq (1-\nu)^{\phi'}\bar{\Delta}$. For every $e$ that gets oriented in or before phase $\phi$, we therefore have $\deg_G(e)\geq (1-\nu)^{\phi}\bar{\Delta}$. This means that for every edge $e$ that is oriented by the end of phase $\phi$, we in particular have
  \begin{equation}
    \label{eq:ke_bound}
    k_e = \left\lceil\frac{\nu}{1-\nu}\cdot \deg_G(e)\right\rceil \geq
    \big\lceil \nu(1-\nu)^{\phi-1}\bar{\Delta}\big\rceil.
  \end{equation}
  For the induction step of our main induction (on $\phi$), we distinguish $4$ cases:

  \paragraph{Edge $e$ gets newly oriented in phase $\phi$:} Edge $e$ gets oriented from $u$ to $v$ in phase $\phi$ if $x_v(\phi-1)-x_u(\phi-1)\leq \eta_e$ and it gets oriented from $v$ to $u$ otherwise (i.e., if $x_u(\phi-1)-x_v(\phi-1)<-\eta_e$). Hence, if $e$ is oriented from $u$ to $v$, \Cref{eq:xv_bounds} implies that $x_v(\phi)-x_u(\phi)\leq \eta_e + k_{\phi}$ and if $e$ is oriented from $v$ to $e$, \Cref{eq:xv_bounds} implies that $x_u(\phi)-x_v(\phi)\leq -\eta_e + k_{\phi}$. In both cases, the claim of the lemma follows together with 
  \[
    k_{\phi}=\big\lceil\nu\cdot(1-\nu)^{\phi-1}\bar{\Delta}\big\rceil
    \stackrel{\eqref{eq:ke_bound}}{\leq} k_e.
  \]
  
  \paragraph{Edge $e$ does not participate in the token dropping game of phase $\phi$:}
  For the remaining cases, we assume that $e$ was first oriented prior to phase $\phi$. 
  Since $e$ does not to participate in the token dropping in phase $\phi$, the orientation of $e$ does not change in phase $\phi$. For $e$ not to participate in the token dropping in phase $\phi$, we then also must have that either $e$ is oriented from $u$ to $v$ and $x_v(\phi-1)-x_u(\phi-1)\leq\eta_e$ or $e$ is oriented from $v$ to $u$ and $x_u(\phi-1)-x_v(\phi-1)\leq -\eta_e$. In both cases, the claim of the lemma follows directly by combining this with \Cref{eq:xv_bounds,eq:ke_bound}.

  \paragraph{A token is moved over $e$ in the token dropping game of phase $\phi$:} If prior to running the token dropping game of phase $\phi$, $e$ is oriented from $u$ to $v$, then $x_v(\phi-1)-x_u(\phi-1)>\eta_e$ and otherwise $x_u(\phi-1)-x_v(\phi-1)>-\eta_e$. Since a token is moved over $e$ in the token dropping game of phase $\phi$, the orientation of $e$ is switched in phase $\phi$. Hence, if before running token dropping in phase $\phi$, the edge $e$ is oriented from $u$ to $v$, at the end of the phase, $e$ is oriented from $v$ to $u$ and 
  \[
  x_u(\phi) - x_v(\phi) \stackrel{\eqref{eq:xv_bounds},\eqref{eq:ke_bound}}{\leq} 
  \underbrace{x_u(\phi-1) - x_v(\phi-1)}_{x_v(\phi-1)-x_u(\phi-1)>\eta_e} + k_e
  < -\eta_e + k_e.
  \]
  If before the token dropping instance of phase $\phi$, $e$ is oriented from $v$ to $u$, at the end of the phase, the edge is oriented from $u$ to $v$ and
  \[
  x_v(\phi) - x_u(\phi) \stackrel{\eqref{eq:xv_bounds},\eqref{eq:ke_bound}}{\leq} 
  \underbrace{x_v(\phi-1) - x_u(\phi-1)}_{x_u(\phi-1)-x_v(\phi-1)>-\eta_e} + k_e
  < \eta_e + k_e.
  \]
  The claim of the lemma therefore also follows in this case.

  \paragraph{No token is moved over $e$ in the token dropping game of phase $\phi$:} The last case to consider is the case where $e$ participates in the token dropping game, but where no token is moved over the edge. First assume that $e$ is oriented from $u$ to $v$, that is, the edge is directed from $v$ to $u$ in the token dropping game. Let $\tau_u$ and $\tau_v$ denote the number of tokens at node $u$ and $v$ at the end of the token dropping game instance of phase $\phi$. 
  Observe that
  \begin{align*}
    \phalpha{v} &= \max\set{1,\frac{1}{4}\cdot \frac{\nu^2}{\ln{\bar{\Delta}}}\cdot (\degmin(v)+1)}\\
    &\ge \max\set{1,\frac{1}{4}\cdot \frac{\nu^2}{\ln{\bar{\Delta}}}\cdot (1-\nu)^{\phi-1}\bar{\Delta}}\\
    &\ge  \max\set{1, \left\lfloor \frac{1}{16}\cdot \frac{\nu^6}{\ln^3\bar{\Delta}}\cdot (1-\nu)^{\phi-1}\bar{\Delta}\right\rfloor} = \delta_{\phi},
  \end{align*}
  where the first inequality holds because, for every $e$ that gets oriented before phase $\phi$, we have $\deg_G(e)\geq (1-\nu)^{\phi-1}\bar{\Delta}$.
  Hence, $\phalpha{v} \ge \delta_{\phi}$, and since no token is moved from $v$ to $u$ in the token dropping game, \Cref{thm:tokenbound} implies that
  \begin{eqnarray*}
    \tau_v - \tau_u
    & \leq &
             2(\phalpha{u} + \phalpha{v}) + \left(\frac{\deg_{F_{<\phi}}(u)\cdot \deg_{F_{<\phi}}(v)}{\phalpha{u}\cdot\phalpha{v}} + \frac{\deg_{F_{<\phi}}(u)}{\phalpha{u}} +
             \frac{\deg_{F_{<\phi}}(v)}{\phalpha{v}}\right)\cdot\delta_{\phi}\\
    & \leq &  \max\set{4,\frac{\nu^2}{2\ln\bar{\Delta}}\cdot (\degmin(u)+\degmin(v)+2)}+
             \left(\frac{16\ln^2\bar{\Delta}}{\nu^4} + \frac{8\ln\bar{\Delta}}{\nu^2}\right)\cdot \delta_{\phi}\\
    & \leq & \frac{\nu^2}{\ln\bar{\Delta}}\cdot\deg_G(e) + 4 +
             \left(\frac{\nu^2}{\ln\bar{\Delta}} + \frac{\nu^4}{2\ln^2\bar{\Delta}}\right)\cdot  (1-\nu)^{\phi-1}\bar{\Delta}
             + \frac{16\ln^2\bar{\Delta}}{\nu^4} + \frac{8\ln\bar{\Delta}}{\nu^2}\\
    & \leq & \frac{5\nu^2}{2\ln\bar{\Delta}}\cdot\deg_G(e) +
             28\cdot \frac{\ln^2\bar{\Delta}}{\nu^4}
  \end{eqnarray*}
  The second inequality follows from plugging in the definition of $\phalpha{u}$ and $\phalpha{v}$ (\Cref{eq:degmin_alpha}) and the fact that for all nodes $v\in V$, $\deg_{F_{<\phi}}(v)\leq \degmin(v)+1$. The third inequality follows from using that $\nu^2/\ln\bar{\Delta}\leq 1$ and from the definition of $\delta_{\phi}$ (\Cref{eq:delta_def}). In the fourth inequality, we use that since edge $e$ was first oriented in a phase before $\phi$, we must have $\deg_G(e)\geq (1-\nu)^{\phi-1}\bar{\Delta}$ and again by using that $\nu^2/\ln\bar{\Delta}\leq 1$.

  The induction step in the induction over the number of phases and thus the claim of the lemma now follows directly from the definitions of $k_e$ and of $\xi_e$ (\Cref{eq:k_xi_edge}).  If the edge $e$ is oriented from $v$ to $u$, then the statement holds for the same reasons.
\end{proof}

The following theorem can now be proven by \Cref{lemma:degreereduction} and by plugging together \Cref{lemma:orientationgrowth}, \Cref{eq:k_xi_edge}, and \Cref{thm:tokenbound}. The proof appears in \Cref{app:deferreddefective}.

\begin{theorem}\label{thm:orientation}
  Assume that we are given a bipartite graph $G=(U\dot{\cup}V, E)$, a value $\eps>0$, and edge parameters $\eta_e\in\mathbb{R}$ for all edges $e\in E$. There exists a constant $C>0$ such that if $\eps\leq 1$, there is an $O\big(\frac{\log^4\Delta}{\eps^6}\big)$-round distributed algorithm to compute a generalized $(\eps,\beta)$-balanced orientation of $G$ w.r.t.\ the edge parameters $\eta_e$, where $\beta=C\cdot \frac{\ln^3\bar{\Delta}}{\eps^5}$.
\end{theorem}

By combining \Cref{thm:orientation} and \Cref{lemma:orientation2defective}, we obtain the following.
\begin{corollary}\label{cor:generalized-defective-2-edge-coloring}
    Let $\varepsilon \le 1$. The generalized $(1+\varepsilon,\beta)$-relaxed defective $2$-edge coloring can be solved deterministically in the \CONGEST model in $O\big(\frac{\log^4 \Delta}{\varepsilon^6}\big)$ rounds, for $\beta = O\big(\frac{\log^3 \Delta}{\varepsilon^5}\big)$.
\end{corollary}


\section{$\boldsymbol{O(\Delta)}$-Edge Coloring in the CONGEST Model}
\label{sec:econgest}

In this section, we present an algorithm for solving the $O(\Delta)$-edge coloring problem in the \CONGEST model. We start by showing how to solve the problem on bipartite $2$-colored graphs. We will later show how to remove this restriction.

\begin{lemma}\label{thm:bipartitecongest}
  The $(2+\varepsilon)\Delta$-edge coloring problem can be solved in $O\big(\frac{\log^{11} \Delta}{\varepsilon^6}\big)$ deterministic rounds in the \CONGEST model in bipartite $2$-colored graphs, for any $1 \ge \varepsilon > 0$.
\end{lemma}
\begin{proof}[Proof Sketch]
  The high-level idea of the proof is the following. In \Cref{cor:generalized-defective-2-edge-coloring}, we show that a $(1+\eps', \beta)$-relaxed defective $2$-edge coloring can be solved in $\poly(\log(\Delta)/\eps')$ time in the \CONGEST model (for $\beta = O\big(\frac{\log^3 \Delta}{\varepsilon'^5}\big)$). As long as the maximum edge degree $\bar{\Delta}$ is sufficiently larger than $\beta/\eps'$, we can therefore compute a $2$-defective edge coloring for which the maximum defect is only by a $(1+\eps')$ factor larger than $\bar{\Delta}/2$. Choosing $\eps'\leq c\cdot \eps/\log\Delta$ for a sufficiently small constant $c$ and an integer $k\geq 1$, as long as $k\gg \beta/\eps'$, we can therefore recursively compute a defective $2^k$-edge coloring such that the maximum edge defect is at most $(1+\eps/2)\bar{\Delta}/2^k$. The required $(2+\eps)\Delta$-coloring of the given bipartite graph is then obtained by using $2^k$ disjoint color ranges for each of the $2^k$ graphs of maximum degree $(1+\eps/2)\bar{\Delta}/2^k$ that we get from the recursive defective coloring. A full proof of the claim appears on \Cref{app:deferredcongest}.
\end{proof}

In order to solve the problem in general graphs, we make use of the following lemma, that follows directly from results presented in~\cite{barenboim14distributed} (a proof is also shown in \Cref{app:deferredcongest}).
\begin{lemma}\label{lem:d4c}
	The $(\varepsilon \Delta + \lfloor \Delta / 2\rfloor)$-defective vertex $4$-coloring problem, given an $O(\Delta^2)$-vertex coloring, can be solved in $O(1 / \varepsilon^2)$ rounds in the \CONGEST model.
\end{lemma}

We are now ready to present our \CONGEST algorithm for $O(\Delta)$-edge coloring on general graphs. \Cref{thm:congest} follows directly from the following theorem.
\begin{theorem}\label{thm:maincongestdetailed}
  The $(8+\varepsilon)\Delta$-edge coloring problem can be solved in
  $O\big(\frac{\log^{12} \Delta}{\varepsilon^6} + \log^* n\big)$ deterministic rounds in the
  \CONGEST model, for any small enough constant $\varepsilon>0$.
\end{theorem}
\begin{proof}[Proof Sketch]
  We start by computing an initial $O(\Delta^2)$-vertex coloring, which can be done in $O(\log^* n)$ rounds.
	Let $\varepsilon_1$ be a parameter that we choose appropriately
	We apply \Cref{lem:d4c} with parameter $\varepsilon_1$ for $4$-coloring the nodes of the graph with colors in $\{1,2,3,4\}$. Then, let $G_1$ be the graph induced by edges $\{u,v\}$ satisfying that the color of $u$ is either $1$ or $2$, and the color of $v$ is either $3$ or $4$. This graph is clearly bipartite, and nodes know their side of the bipartition. Hence, we can apply \Cref{thm:bipartitecongest} to color the edges of this graph by using at most $(2+\varepsilon_2)\Delta$ colors, for some appropriate parameter $\varepsilon_2$. We then do the same in the bipartite graph induced by the edges that go from colors $\set{1,3}$ to colors $\set{2,4}$. For the edges that are colored so far, we have now used $(4+2\eps_2)\Delta$ colors. All the remaining uncolored edges are now between nodes of the same color and hence the degree of them is at most $(1/2+\eps_1)\Delta$ and thus close to $\Delta/2$. To color the rest of the graph, we now recurse. As the maximum degree (almost) halves in each step, the total number of colors will be (essentially) twice the number of colors we have used so far and thus $(8+O(\eps_1))\Delta$. A full proof appears in \Cref{app:deferredcongest}.
\end{proof}

\section{$\boldsymbol{(2\Delta-1)}$-Edge Coloring in the LOCAL Model}
\label{sec:eclocal}

In this section, we give an overview of our algorithm that computes a $(\mathit{degree}+1)$-list edge coloring and thus as a special case a $(2\Delta-1)$-edge coloring. For the detailed formal arguments, we refer to \Cref{app:eclocal}.

\paragraph{Coloring $2$-Colored Bipartite Graphs.} We first again consider the case of computing a somewhat relaxed coloring of a $2$-colored bipartite graph $G=(U\cup V, E)$. For the high level description here, assume that $G$ has maximum degree $\Delta$ and that every edge $e\in E$ has a list $L_e$ of size $|L_e|\geq 2\deg_G(e)$. Assume further that $L_e\subseteq\set{1,\dots,C}$, i.e., all colors come from a global space of $C$ colors. If all nodes have access to the same colors (as in the algorithm of \Cref{sec:econgest}), we can use defective $2$-colorings to split the graph into two parts such that the maximum edge degree in each part is approximately halved compared to the original maximum degree. However, in the case of list coloring, we cannot do this because in a local distributed way, we cannot split the colors into two parts such that each node can keep half of its colors. We instead have to adapt a method that was introduced in \cite{Kuhn20} and also used in \cite{BalliuKO20}. First, the global space of colors $\set{1,\dots,C}$ is split into two approximately equal parts. For this, let us call the colors $\set{1,\dots,\lfloor C/2\rfloor}$ red and the remaining colors blue. We want to color the edges red and blue such that afterwards, the red edges $e$ only keep the red colors in their list $L_e$ and the blue edges only keep their blue colors. In this way, the two parts are independent of each other and can be colored in parallel. Note however that because the lists $L_e$ are arbitrary subsets of $\set{1,\dots,C}$, the list $L_e$ of an edge $e$ can consist of an arbitrary division into red and blue colors. For an edge $e\in E$, let $\lambda_e$ be the fraction of red colors in its list, i.e., $|L_e\cap\set{1,\dots,\lfloor C/2\rfloor}|=\lambda_e|L_e|$. If $e$ chooses to be red, its list shrinks by a factor $\lambda_e$ and we therefore also want to shrink $e$'s degree by at least (approximately) a factor $\lambda_e$ and if $e$ chooses to be blue, the $e$'s degree has to shrink by at least (approximately) a factor $1-\lambda_e$. If the degree of $e$ is sufficiently large, we can achieve this by computing a generalized defective $2$-edge coloring as defined in \Cref{def:defective2coloring} and we can use \Cref{cor:generalized-defective-2-edge-coloring} to compute such a $2$-coloring efficiently. The goal therefore is to recursively split the global color space into two parts and always use \Cref{cor:generalized-defective-2-edge-coloring} to split the edges such that the degree-to-list size ratio grows by at most a factor $1+o(1)$. This essentially works as in \Cref{sec:econgest}. There is only one additional small issue that we have to take care of. Since the ratio $\lambda_e$ can be an arbitrary value between $0$ and $1$, we have no control over the minimum and maximum edge degree or list size in the graph. \Cref{cor:generalized-defective-2-edge-coloring} however only gives good guarantees for edges that have a sufficiently large degree. As soon as an edge $e$ has degree at most some $\poly\log\Delta$, we therefore do not further split the color space of $e$ recursively. Edge $e$ then waits until all neighboring edges that are split further have been colored and afterwards, $e$ only has a small uncolored degree and can therefore be colored greedily by a standard edge coloring algorithm.

\paragraph{Degree+1 List Edge Coloring on General Graphs.} 
Given a $(\mathit{degree}+1)$-list coloring instance of an arbitrary graph $G$, we start by computing a $\poly(\Delta)$-vertex coloring, which can be done in $O(\log^* n)$ rounds. 

Then, we compute a defective $O(1)$-vertex coloring of $G$, where each node has defect at most $\Delta/c$ for a sufficiently large constant $c$. By using an algorithm from \cite{barenboim14distributed}, this can be done in $O(\log^* \Delta)$ time by exploiting the precomputed $\poly(\Delta)$-vertex coloring.

We then sequentially iterate through all possible pairs of colors $(a,b)$, and we consider the graph induced by edges with one endpoint of color $a$ and the other endpoint of color $b$. This graph is clearly bipartite, but we cannot directly apply the previously described algorithm, because the lists of the edges may be too small, compared to their degree. In order to apply the algorithm anyways, we use a method that has already been used in \cite{fraigniaud16local,Kuhn20,BalliuKO20} and that allows us to use the previously described algorithm to \emph{partially} color the graph. 

Like this, we can reduce the uncolored degree of each edge by a constant factor even in a $(\mathit{degree}+1)$-list coloring instance, and repeating $O(\log\Delta)$ times allows to color all the edges and to fully solve the given  $(\mathit{degree}+1)$-list coloring problem.


\bibliographystyle{ACM-Reference-Format}
\bibliography{references}

\appendix

\section{Deferred Proofs of Section \ref{sec:tokendropping}}
\label{app:deferredtoken}

\begin{proof}[\bf Proof of \Cref{thm:tokenbound}]
    The algorithm is run for $\lfloor k/\delta\rfloor -1$ phases, and
  since each phase can clearly be implemented in $O(1)$ rounds, the
  time complexity of the algorithm is therefore $O(k/\delta)$. The
  algorithm terminates after $T=\lfloor k/\delta\rfloor -1$
  rounds. The number of tokens $\tau(v)$ of a node $v$ at the end of
  the algorithm is therefore $\tau(v) = x_v(T)+ y_v(T)$. By
  \Cref{lemma:tokenbounds}, we therefore have $\tau(v)\leq k$.  We can
  further upper bound $\tau(u)-\tau(v)$ by
  \begin{equation}
    \label{eq:tokendifference}
    \tau(u)-\tau(v) \leq x_u(T) + y_u(T) - y_v(T).
  \end{equation}
  By \Cref{lemma:tokenbounds}, we have
  \begin{eqnarray*}
    x_u(T) & \leq & \max\set{2\alpha_u,k-T\delta}\\
           & =& \max\set{2\alpha_u, k-\left(\left\lfloor \frac{k}{\delta}\right\rfloor-1\right) \cdot \delta}\\
           & \leq &
                    \max\set{2\alpha_u,k-\left(\frac{k}{\delta} - 2\right)\cdot \delta}\\
           & = & \max\set{2\alpha_u,2\delta}\ \stackrel{(\alpha_u\geq\delta)}{\leq} 2\alpha_u.
  \end{eqnarray*}
  The bound on $\tau(u)-\tau(v)$ now therefore directly follows together with \cref{eq:tokendifference} and \Cref{lemma:passiveslackbound}.
\end{proof}

\section{Deferred Proofs of Section \ref{sec:gen-def-2-ec}}
\label{app:deferreddefective}

\begin{proof}[\bf Proof of \Cref{lemma:orientation2defective}]
  Let $d_e$ be the defect of edge $e$, i.e., $d_e$ denotes the total
  number of neighboring edges of $e$ that have the same color as $e$.

  Consider some edge $e=(u,v)$ for $u\in U$ and $v\in V$. We first
  assume that $e$ is oriented from $u$ to $v$ and thus $e$ is red. The
  total number of neighboring red edges of $e$ is then equal to the
  number of edges $e'\neq e$ that are oriented out of $u$ plus the
  number of edges $e''\neq e$ that are oriented into $v$, i.e.,
  \begin{eqnarray*}
    d_e
    & = & \underbrace{\deg_G(u) - x_u}_{\text{\# edges oriented out of $u$}} - 1 + \underbrace{x_v}_{\text{\# edges oriented into $v$}} - 1\\
    & \leq & \deg_G(u) + \eta_e + 1 + \frac{\eps}{2}\cdot\deg_G(e) +\beta - 2\\
    & \stackrel{\eqref{eq:def_eta}}{=} & -2\lambda_e +\lambda_e\deg_G(u)+\lambda_e\deg_G(v)   +
                                         \eps\lambda_e\cdot\deg_G(e) + (2\lambda_e-1)\beta + \beta\\
    & = & (1+\eps)\cdot\lambda_e\cdot\deg_G(e) + 2\lambda_e\beta.
  \end{eqnarray*}
  The first inequality follows because if $e$ is red, we have $x_v-x_u\leq \eta_e + 1 + \frac{\eps}{2}\cdot\deg_G(e)+\beta$.
  If $e=(u,v)$ is oriented from $v$ to $u$ and thus $e$ is blue, we similarly obtain
  \begin{eqnarray*}
    d_e
    & = & \underbrace{\deg_G(v) - x_v}_{\text{\# edges oriented out of $v$}} - 1 + \underbrace{x_u}_{\text{\# edges oriented into $u$}} - 1\\
    & \leq & \deg_G(v) - \eta_e + 1 + \frac{\eps}{2}\cdot\deg_G(e) +\beta - 2\\
    & \stackrel{\eqref{eq:def_eta}}{=} & 2\lambda_e +(1-\lambda_e)\deg_G(u)+(1-\lambda_e)\deg_G(v)  -
                                         \eps\lambda_e\cdot\deg_G(e) +\eps\cdot\deg_G(e)+(2-2\lambda_e)\beta-2\\
    & = & (1+\eps)\cdot(1-\lambda_e)\cdot\deg_G(e)+2(1-\lambda_e)\beta.
  \end{eqnarray*}
  The first inequality follows because if $e$ is blue, we have $x_u-x_v\leq -\eta_e+1+\frac{\eps}{2}\cdot\deg_G(e)+\beta$.
\end{proof}

\begin{proof}[\bf Proof of \Cref{lemma:degreereduction}]
  We prove the statement by induction on $\phi$. 
  First of all, note that for every edge $e$ and every $\phi\geq 1$, we have $d(e,\phi)\leq d(e,\phi-1)$. For every edge $e$ with $d(e,\phi-1)\leq (1-\nu)^{\phi}\bar{\Delta}$, we therefore clearly also have $d(e,\phi)\leq (1-\nu)^{\phi}\bar{\Delta}$. It therefore suffices to consider edges $e$ for which $d(e,\phi-1)>(1-\nu)^{\phi}$ and edges $e$ that are not already oriented at the beginning of phase $\phi$. Note that this is exactly the set $E_{\phi}$ as defined in step 1 of the above algorithm. Every edge $e=\{u,v\}\in E_{\phi}$ proposes to one of its two nodes. W.l.o.g., assume that $e$ proposes to $v$. Either $v$ accepts $e$'s proposal and in this case, $e$ is oriented after phase $\phi$ (and thus, we do not need to show that $d(e,\phi)$ is bounded). Otherwise, $v$ accepts $k_{\phi}=\big\lceil\nu(1-\nu)^{\phi-1}\bar{\Delta}\big\rceil$ other proposals. Node $v$ is therefore incident to $k_{\phi}$ edges that are unoriented at the beginning of phase $\phi$ and that become oriented in phase $\phi$. The number of unoriented incident edges of $v$ and thus the number of neighboring unoriented edges of $e$ therefore decreases by at least $k_{\phi}$ in phase $\phi$ and thus,
  \[
    d(e,\phi) \leq d(e,\phi-1) - k_{\phi} \leq (1-\nu)^{\phi-1}\bar{\Delta} - k_{\phi}
    \leq (1-\nu)^{\phi-1}\bar{\Delta}-\nu(1-\nu)^{\phi-1}\bar{\Delta} = (1-\nu)^{\phi}\bar{\Delta}.
  \]
  The first inequality follows from the induction hypothesis or for $\phi=1$ from $d(e,0)\leq \bar{\Delta}$.
\end{proof}

\begin{proof}[\bf Proof of \Cref{thm:orientation}]
  By \Cref{lemma:degreereduction}, after $\phi$ phases, the maximum edge degree of the subgraph induced by the unoriented edges is at most $(1-\nu)^{\phi}\bar{\Delta}$. This in particular implies that after $\phi$ phases, every node $v\in V$ has at most $1+(1-\nu)^{\phi}\bar{\Delta}$ incident unoriented edges. After $\hat{\phi}=\ln(\bar{\Delta})/\ln(1/(1-\nu)) - O(1)=O(\log(\Delta)/\nu)$ phases, each node therefore has at most $O(1)$ incident unoriented edges. At this point, we can stop the above algorithm and just orient the remaining unoriented edges arbitrarily. This only affects the number of incoming edges at each node by an additive constant. Let us consider the state after $\hat{\phi}$ phases. Consider some edge $e=(u,v)$ and assume first that $e$ is oriented from $u$ to $v$ after phase $\hat{\phi}$. By \Cref{lemma:orientationgrowth}, we then have
  \begin{eqnarray*}
    x_v(\hat{\phi}) - x_u(\hat{\phi})
    & \leq & \eta_e + k_e + \hat{\phi}\cdot\xi_e\\
    & \leq & \eta_e + k_e + \frac{\ln\bar{\Delta}}{\ln\big(\frac{1}{1-\nu}\big)}\cdot
             \left(\frac{5}{2}\cdot\frac{\nu}{\ln\bar{\Delta}}\cdot k_e + 28\cdot\frac{\ln^2\bar{\Delta}}{\nu^4}\right)\\
    & \leq & \eta_e + \frac{7}{2}\cdot\left(\frac{\nu}{1-\nu}\cdot\deg_G(e) + 1\right) +
             28\cdot \frac{\ln^3\bar{\Delta}}{\nu^5}\\
    & \leq & \eta_e + 4\nu\cdot\deg_G(e) + \frac{7}{2} + 28\cdot \frac{\ln^3\bar{\Delta}}{\nu^5}.
  \end{eqnarray*}
  The second inequality follows from \eqref{eq:k_xi_edge} and the
  third inequality follows from \eqref{eq:k_xi_edge} and the fact that
  for $\nu\in(0,1]$, $\nu\leq \ln\big(\frac{1}{1-\nu}\big)$. The last
  inequality follows because $\nu\leq 1/8$ (cf.\ \Cref{eq:def_nu}). By setting $\eps = 8\nu$,
  we therefore get
  $x_v(\hat{\phi}) - x_u(\hat{\phi})\leq \eta_e + \frac{\eps}{2}\cdot\deg_G(e) +
  C\cdot\frac{\ln^3\bar{\Delta}}{\eps^5}$ for some constant $C>0$. In
  the same way, if $e$ is oriented from $v$ to $u$, we obtain that
  $x_u(\hat{\phi}) - x_v(\hat{\phi})\leq -\eta_e + \frac{\eps}{2}\cdot\deg_G(e)
  + C\cdot\frac{\ln^3\bar{\Delta}}{\eps^5}$ for some constant $C>0$.

  It remains to bound the round complexity of the algorithm. By \Cref{thm:tokenbound}, the time required for phase $\phi$ is
  \[
    O\left(\frac{k_\phi}{\delta_\phi}\right) =
    O\left(\frac{\nu(1-\nu)^{\phi-1}\bar{\Delta}}{\frac{\nu^6}{\ln^3\bar{\Delta}}\cdot (1-\nu)^{\phi-1}\bar{\Delta}}\right) =
    O\left(\frac{\log^3\Delta}{\eps^5}\right).
  \]
  The claimed time complexity now follows because the number of phases
  is
  $O\big(\frac{\log\bar\Delta}{\nu}\big)=O\big(\frac{\log\Delta}{\eps}\big)$.
\end{proof}

\section{Deferred Proofs of Section \ref{sec:econgest}}
\label{app:deferredcongest}

\begin{proof}[\bf Proof of \Cref{thm:bipartitecongest}]
  We describe a procedure that, given a graph $G$ of maximum degree
  $\Delta$ and maximum edge degree $\bar{\Delta}$, is able to split it
  into two disjoint subgraphs, each of maximum edge degree
  $\bar{\Delta} \cdot (1+\chi)/2 + \beta$, for some parameter $\chi$
  satisfying $c \cdot \varepsilon / \log \Delta \le \chi\le 1/2$ (for
  some constant $c$) that we will fix later, and
  $\beta = O\big(\frac{\log^3\Delta}{\chi^5}\big)$. Note that the assumption
  on $\chi$ implies that
  $\beta = O\big(\frac{\log^8 \Delta}{c^5 \varepsilon^5}\big)$. We will then
  apply this procedure recursively in both subgraphs for many times,
  in order to obtain many subgraphs of small maximum edge degree. We
  will then color the edges of each subgraph independently in
  parallel, and the final color of an edge $e$ in $G$ will be given by
  its color in the subgraph $G'$ containing $e$, combined with the
  index of $G'$.
	
  The procedure works as follows. For each edge $e$, we fix
  $\lambda_e = 1/2$ and we apply
  \Cref{cor:generalized-defective-2-edge-coloring} with parameter
  $\chi$. We obtain a defective $2$-edge coloring of $G$, satisfying
  that, for each edge $e$,
  $\deg_{G'}(e) \le \deg_{G}(e) \cdot (1+\chi)/2 + \beta$, where $G'$
  is the graph induced by edges with the same color of $e$. Hence, we
  decomposed our graph $G$ into two edge-disjoint subgraphs $G_1$ and
  $G_2$, both having maximum edge degree at most
  $\bar{\Delta} \cdot (1+\chi)/2 + \beta$.
	
  By applying this procedure recursively $k$ times, we get $2^k$
  edge-disjoint subgraphs, each of maximum edge degree
  \begin{equation*}
    d \le \bar{\Delta} \cdot \left(\frac{1+\chi}{2}\right)^k + \beta\sum_{i=0}^{k-1}\left(\frac{1+\chi}{2}\right)^i
    \le \bar{\Delta} \cdot \left(\frac{1+\chi}{2}\right)^k + \frac{2\beta}{1-\chi}
    \le \bar{\Delta} \cdot \left(\frac{1+\chi}{2}\right)^k + 4\beta. 
  \end{equation*}
	
  Each subgraph can then be edge-colored with $d+1$ colors. We can
  characterize each edge by a tuple $(\mathsf{vec},\mathsf{col})$,
  where $\mathsf{vec}$ is a vector in $\{0,1\}^k$ identifying, for
  each step of the recursion, which color the edge obtained, and
  $\mathsf{col}$ is the color in $\{0,\ldots, d\}$ obtained in the
  last step. Note that if two edges are neighbors in $G$, they must
  have different tuples, because either they ended up in different
  subgraphs, or they obtained different colors in the last coloring
  step. Hence, each tuple can be seen as a color from a palette of
  size
  \[ p = (1+d) \cdot 2^k \le \bar{\Delta}(1+\chi)^k + 2^k(1 + 4\beta)
    \le \bar{\Delta}(1+\chi)^k + 2^k \cdot 5\beta.
  \] 
  By fixing $k = \lfloor \ln(1+\varepsilon/4) / \chi \rfloor$, we obtain 
  \[
    \bar{\Delta}(1+\chi)^k \le \bar{\Delta}(1+\chi)^{\frac{\ln (1+\varepsilon/4)}{\chi}} \le \bar{\Delta} e^{\ln(1+\varepsilon/4)} = \bar{\Delta}(1+\varepsilon/4),
  \]
  and
  \[
    2^k \cdot 5\beta \le (1 + \varepsilon/4)^{\frac{\ln 2}{\chi}} \cdot 5 \beta \le  (1 + \varepsilon/4)^{\frac{\ln 2}{\chi}}  \cdot c' \frac{\log^{8} \Delta}{c^5 \varepsilon^5}
  \]
  for some constant $c'$.
  We now fix $\chi$ to be
  \[
    \chi = \frac{\log(1 + \varepsilon/4)\ln 2}{\log\big(\frac{\varepsilon\bar{\Delta}/4}{c'\log^{8}(\Delta)/(c^5\varepsilon^5)}\big)}.
  \]
  Note that, by taking a small enough constant $c$, we still satisfy
  $\chi \ge c \cdot \varepsilon / \log \Delta$ as required. Also, notice that $1/\chi = O(\frac{1}{\varepsilon}\log \Delta)$.
  We obtain
  \[
    2^k \cdot 5\beta \le \frac{\varepsilon\bar{\Delta}/4}{c'\log^{8}(\Delta)/(c^5\varepsilon^5)} \cdot c'\log^{8}(\Delta)/(c^5\varepsilon^5) = \varepsilon \bar{\Delta} / 4.
  \]
  Hence, $p \le \bar{\Delta}(1+\varepsilon/2)$. Since $\bar{\Delta} \le 2\Delta-2$, we obtain a coloring that uses at most $2 \Delta(1+\varepsilon/2) = (2+\varepsilon)\Delta$ colors, as required.
  
  The total running time is bounded by the number of recursion steps $k$, multiplied by the time required to solve the generalized defective $2$-edge coloring problem with parameter $\chi$, plus the time required to compute the $(d+1)$-coloring in the last step.
  Note that \[
    2^k \ge \frac{1}{2} \cdot (1 + \varepsilon/4)^{\frac{\ln 2}{\chi}} = \frac{\varepsilon\bar{\Delta}/4}{2 c'\log^{8}(\Delta)/(c^5\varepsilon^5)},
  \]
  and hence $d$ is upper bounded by
  \[
    d \le \bar{\Delta}(1+\chi)^k / 2^k + 4 \beta \le \frac{\bar{\Delta}(1+\varepsilon/4)}{\varepsilon\bar{\Delta}/4} 2c'\log^{8}(\Delta)/(c^5\varepsilon^5) + 4\beta = O\left(\frac{\log^{8}\Delta}{\varepsilon^6} + \beta\right).
  \]
  In the \CONGEST model, we can compute a $(d+1)$-edge coloring in $O(d)$ rounds \cite{BEG18} (note that a $\log^* n$ dependency is not necessary in our case, since, in \cite{BEG18}, it is only spent to compute an initial $O(d^2)$-edge coloring, which can be done in $O(1)$ rounds if we are given a $2$-vertex coloring). Hence, the last phase costs $O\big(\frac{\log^{8}\Delta}{\varepsilon^6} + \beta\big)$ rounds.
  
  Summarizing, each coloring phase costs
  $O\big(\frac{\log^4 \Delta}{\chi^6}\big)$ rounds, and we perform
  $k = O\big(\frac{\varepsilon}{\chi}\big)$ phases. Since
  $1/\chi = O\big(\frac{1}{\varepsilon}\log \Delta\big)$, then this part costs
  $O\big(\frac{\log^4 \Delta }{\chi^6} \cdot \frac{\varepsilon}{\chi}\big) =
  O\big(\frac{\log^4 \Delta }{\chi^7} \cdot \varepsilon\big) =
  O\big(\frac{\log^{11} \Delta }{\varepsilon^6}\big)$ rounds.  The coloring of
  the last phase costs
  $O\big(\frac{\log^{8} \Delta}{\varepsilon^6} + \beta\big) =
  O\big(\frac{\log^{8}\Delta}{\varepsilon^6} + \frac{\log^3
    \Delta}{\chi^5}\big) = O\big(\frac{\log^{8}\Delta}{\varepsilon^6} +
  \frac{\log^8 \Delta}{\varepsilon^5}\big) =
  O\big(\frac{\log^{8}\Delta}{\varepsilon^6}\big)$ rounds.  Hence, in total,
  we spend
  $O\big(\frac{\log^{11} \Delta}{\varepsilon^6}
  +\frac{\log^{8}\Delta}{\varepsilon^6}\big) = O\big(\frac{\log^{11}
    \Delta}{\varepsilon^6}\big)$ rounds, proving the theorem.
\end{proof}

\begin{proof}[\bf Proof of \Cref{thm:maincongestdetailed}]
	We start by computing an initial $O(\Delta^2)$-vertex coloring, which can be done in $O(\log^* n)$ rounds.
	Let $\varepsilon_1$ be a parameter to be fixed later.
	We apply \Cref{lem:d4c} with parameter $\varepsilon_1$ for $4$-coloring the nodes of the graph with colors in $\{1,2,3,4\}$. Then, let $G_1$ be the graph induced by edges $\{u,v\}$ satisfying that the color of $u$ is either $1$ or $2$, and the color of $v$ is either $3$ or $4$. This graph is clearly bipartite, and nodes know their side of the bipartition. Hence, we can apply \Cref{thm:bipartitecongest} to color the edges of this graph by using at most $(2+\varepsilon_2)\Delta$ colors, for some parameter $\varepsilon_2$ to be fixed later.
	
	We then consider the graph $G_2$ induced by uncolored edges $\{u,v\}$ satisfying that the color of $u$ is either $1$ or $3$, and the color of $v$ is either $2$ or $4$. Again, we can apply  \Cref{thm:bipartitecongest} to color the edges of this graph by using at most $(2+\varepsilon_2)\Delta$ colors. We obtain a partial coloring of the edges of $G$ with $(4+2\varepsilon_2)\Delta$ colors, where all bichromatic edges are colored.
	
	We apply the same procedure recursively on the graph induced by uncolored edges, for other $k$ times ($k+1$ times in total, indexed by $i \in \set{0,\ldots,k})$. At the end, we color the remaining uncolored edges. Note that, by \Cref{lem:d4c}, the degree of the graph induced by uncolored edges decreases at least by a factor $(1/2 + \varepsilon_1)$ at each step, and hence the maximum degree of the graph before step $i$ is bounded by $d_i = \Delta (1/2+\varepsilon_1)^{i}$.
	The maximum degree induced by uncolored edges after step $k$ is $d_{k+1}$, and hence it can be edge colored with $2d_{k+1}-1$ colors.
	Hence, the total number of colors used by the algorithm is
	\begin{align*}
	  c &= 2d_{k+1} -1 + \sum_{i=0}^{k} (4+2\varepsilon_2) d_i \\
	    &< 2\Delta (1/2+\varepsilon_1)^{k+1} + \Delta (4+2\varepsilon_2) \sum_{i=0}^{k}  (1/2 + \varepsilon_1)^i\\ 
		&=  \Delta \left( 2(1/2+\varepsilon_1)^{k+1}  + \frac{4+2\varepsilon_2}{1/2-\varepsilon_1}(1 - (1/2+\varepsilon_1)^{k+1}) \right) \\
		&=  \Delta \left( \frac{4+2\varepsilon_2}{1/2-\varepsilon_1} - (\frac{4+2\varepsilon_2}{1/2-\varepsilon_1}-2)(1/2+\varepsilon_1)^{k+1} \right) \\
		&\le  \Delta \left( \frac{4+2\varepsilon_2}{1/2-\varepsilon_1} - \frac{6}{2^{k+1}} \right) \\
		&=\Delta(8 + O(\varepsilon_1 + \varepsilon_2)).
	\end{align*}
	If we fix $k = \lfloor\log \Delta \rfloor - 1$ and $\varepsilon_1 = 1 / (2k)$, then the final graph has maximum degree
	\[
	d_{k+1} = \Delta (1/2+\varepsilon_1)^{k+1} \le (1 + 1/\log \Delta)^{\log \Delta} = O(1),
	\] 
	and thus can be edge colored with the required amount of colors,  $2d_{k+1}-1$, in $O(\log^* n)$ rounds. 
	
	The total running time is hence $O(T_{\mathrm{initial}} + k \cdot (T_{\mathrm{defective}} + T_{\mathrm{bipartite}}) + T_{\mathrm{final}})$, where $T_{\mathrm{initial}} = O(\log^* n)$ is the time required to compute an initial $O(\Delta^2)$-vertex coloring, $k+1$ is the number of recursive steps, $T_{\mathrm{defective}} = O(1 / \varepsilon_1^2)$ is the time required to apply \Cref{lem:d4c}, $T_{\mathrm{bipartite}} = O\big(\frac{\log^{11} \Delta}{\varepsilon_2^6}\big)$ is the time required to apply \Cref{thm:bipartitecongest}, and $T_{\mathrm{final}} = O(\log^* n)$ is the time required to color the resulting constant degree graph in the last step.
	
	By fixing $\varepsilon_2 = \varepsilon$, we hence obtain a $\Delta(8 + O(\varepsilon + 1 / \log \Delta)) = \Delta(8 + O(\varepsilon))$ edge coloring in $O\big(\log^* n + \log \Delta \cdot  \big(\log^2 \Delta + \frac{\log^{11} \Delta}{\varepsilon^6}\big) + \log^* n\big) = O\big(\frac{\log^{12} \Delta}{\varepsilon^6} + \log^* n\big)$ rounds.
      \end{proof}

\begin{proof}[Proof of \Cref{lem:d4c}]
	We start by computing a $p$-defective $O((\Delta/p)^2)$-vertex coloring, for $p=\varepsilon \Delta$, which can be done in $O(1)$ rounds given an $O(\Delta^2)$-coloring \cite{barenboim14distributed}. Then, we apply the procedure Refine of \cite{barenboim14distributed}, to obtain an $(\varepsilon \Delta + \lfloor \Delta / 2\rfloor)$-defective $4$-coloring in $O(1 / \varepsilon^2)$ rounds.
\end{proof}


\section{$\boldsymbol{(2\Delta-1)}$-Edge Coloring in the LOCAL Model: Formal Details}
\label{app:eclocal}
In this section, we present an algorithm for solving the $(2\Delta-1)$-edge coloring problem in the \LOCAL model.
We first prove that, if the degrees of the edges are not too small, then we can use \Cref{cor:generalized-defective-2-edge-coloring} to split our graph $G$ into two edge disjoint subgraphs $G_1$ and $G_2$, such that the slack does not decrease by much. More precisely, we start by proving the following lemma.
\begin{lemma}\label{lem:splithighdegree}
	Let $G = (V,E)$ be a bipartite $2$-colored graph of maximum degree at most $\Delta$. Let $\varepsilon > 0$ be a parameter.  Assume that each edge has a list $L_e \subseteq \set{C_1,\ldots,C_2}$ assigned to it satisfying $|L_e| > S \cdot d(e)$, where $d(e) \ge \deg(e)$, for some slack parameter $S \ge 1$. Also, assume  $d(e) \ge \beta/\varepsilon$,  where $\beta = c\frac{\log^3 \Delta}{\varepsilon^5}$ for some constant $c>0$.
	Let $C^1 =  \{C_1,\ldots,\lfloor (C_1+C_2)/2\rfloor\}$ and $C^2 =  \{\lfloor (C_1+C_2)/2\rfloor + 1 , \ldots, C_2\}$, and let $L^1_e = L_e \cap C^1$, $L^2_e = L_e \cap C^2$.
	Then, it is possible to split $G$ into two edge-disjoint subgraphs $G_1$ and $G_2$, such that, for each edge $e \in G_i$, $|L^i_e| > S \deg_{G_i}(e) / (1+\varepsilon)^2$, in $O\big(\frac{\log^4 \Delta}{\varepsilon^6}\big)$ rounds.
\end{lemma}
\begin{proof}
	Let $\lambda_e = |L^1_e| / |L_e|$. 
	We apply \Cref{cor:generalized-defective-2-edge-coloring}. Let $G_1$ and $G_2$ be the graphs induced by edges of the first color, and second color, respectively. We obtain that, for each edge $e$ in $G_1$, the degree is
	\begin{align*}
	\deg_{G_1}(e) 
	&\le (1+\varepsilon) \cdot \lambda_e \cdot \deg_G(e) + \lambda_e \beta \\
	&\le (1+\varepsilon) \cdot \lambda_e \cdot d(e) + \lambda_e \cdot \varepsilon \cdot d(e) \\
	&= (1+2\varepsilon) \cdot \lambda_e \cdot d(e)\\
	&\le (1+\varepsilon)^2 \cdot \lambda_e \cdot d(e)\\
	&= (1+\varepsilon)^2 \cdot \frac{|L^1_e|}{|L_e|} \cdot d(e)\\
	& < (1+\varepsilon)^2 \cdot \frac{L^1_e}{S}.
	\end{align*}
	That is, $|L^1_e| > \deg_{G_1}(e) \cdot S / (1+\varepsilon)^2$, hence, in other words, the new slack is at least $S/(1+\varepsilon)^2$.
	Similarly, for each edge $e$ in $G_2$, the degree is 
	\begin{align*}
	\deg_{G_2}(e) 
	&\le (1+\varepsilon) \cdot (1-\lambda_e) \cdot \deg_G(e) + (1-\lambda_e) \beta \\
	&\le (1+\varepsilon) \cdot (1-\lambda_e) \cdot d(e) + (1-\lambda_e) \cdot \varepsilon \cdot d(e) \\
	&= (1+2\varepsilon) \cdot (1-\lambda_e) \cdot d(e)\\
	&\le (1+\varepsilon)^2 \cdot (1-\lambda_e) \cdot d(e)\\
	&= (1+\varepsilon)^2 \cdot \frac{|L^2_e|}{|L_e|} \cdot d(e)\\
	& < (1+\varepsilon)^2 \cdot \frac{L^2_e}{S}.
	\end{align*}
	Hence, also in this case, the new slack is still at least $S/(1+\varepsilon)^2$.
\end{proof}

We now prove that, if we have large enough slack, and we are in a bipartite $2$-colored graphs, then we can solve the list-edge coloring problem.
\begin{lemma}\label{lem:solveslack}
	The problem $P(\bar{\Delta},S,C)$ can be solved in $O(\log^7 C \cdot \log^4 \Delta)$ rounds, if nodes are provided with a $2$-vertex coloring, and $S \ge e^2$.
\end{lemma}
\begin{proof}
	In the following, we fix $\beta = c \frac{\log^3 \Delta}{\varepsilon^5}$ to be an upper bound of the parameter $\beta$ of \Cref{cor:generalized-defective-2-edge-coloring}.
	Our algorithm works as follows.
	\begin{enumerate}
		\item Let $\mathcal{G}^1 = \set{G}$ and $C_G = \set{0,\ldots, C-1}$. At the beginning, by assumption, all edges have a palette of colors $L_e \subseteq C_G$ satisfying $|L_e| > S \deg(e)$. All edges are \emph{active}.
		\item For $i = 1,\ldots,k$, perform the following, in parallel, in each graph $G^i \in \mathcal{G}^i$:\label{phase-2}
			\begin{enumerate}
				\item Let $d_i(e)$ be the number of active neighboring edges of $e$ in $G^i$. Edges satisfying $d_i(e) < \beta / \varepsilon$ become \emph{passive}. Let $\deg_i(e) \le d(e)$ be the number of active neighboring edges of $e$ after this operation.
				\item Apply \Cref{lem:splithighdegree}. The two resulting subgraphs are part of $\mathcal{G}^{i+1}$.
			\end{enumerate}
		\item Color each graph $G^{k+1} \in \mathcal{G}^{k+1}$ in parallel. \label{phase-3}
		\item For $i = k, \ldots, 1$, perform the following, in parallel, in each graph $G^{i} \in \mathcal{G}^{i}$:\label{phase-4}
		\begin{enumerate}
			\item Color the edges of $G^i$ that became passive during phase $i$.
		\end{enumerate}
	\end{enumerate}

We now prove the correctness of the algorithm, and a bound on its running time.
\paragraph{All edges are properly colored.}
Each edge either remains active until the end, and it is colored in \Cref{phase-3}, or it became passive in some phase $i \in \set{1,\ldots,k}$, and in this case it is colored in \Cref{phase-4}. Each time we split a graph into two edge-disjoint subgraphs, we also split the color space into two disjoint subspaces, and hence coloring different subgraphs cannot create conflicts.

\paragraph{Active edges have large slack.}
We prove by induction that, at the beginning of phase $i$ of \Cref{phase-2}, the slack of active edges is at least $S / (1+\varepsilon)^{2i-2}$. The claim trivially holds for $i=1$, that is, at the beginning. Assume that the claim holds in phase $i$, we prove that the claim holds for $i+1$.
Since all edges satisfying $d_i(e) < \beta / \varepsilon$ become passive, then all edges that participate in the application of \Cref{lem:splithighdegree} satisfy $d_i(e) \ge \beta / \varepsilon$, as required. Also, note that the graph induced by active edges satisfies $\deg(e) \le d_i(e)$ for all edges $e$, as required. Hence, \Cref{lem:splithighdegree} guarantees that the graph $G^i$ is split into two edge-disjoint subgraphs satisfying that, for each edge $e$, the new slack of $e$ in its assigned subgraph is at least $S'/(1+\varepsilon)^2$, where $S' = S / (1+\varepsilon)^{2i-2}$ is the original slack by the inductive hypothesis. Hence, the new slack is at least $S / (1+\varepsilon)^{2i}$.

\paragraph{We can color the edges that at the end are still active.}
We fix $k = \lfloor \log C \rfloor$ and $\varepsilon = 1 /\log C$.
The slack of the edges that at the end are still active, that is, the edges that we color in \Cref{phase-3}, is at least
	\[
	 \left(\frac{1}{1 + \varepsilon}\right)^{2k} S \ge \left(\frac{1}{1 + \frac{1}{\log C}}\right)^{2\log C} S \ge S / e^2 \ge 1.
	\]
Hence, each edge satisfies that its list is strictly larger than its degree.

\paragraph{We can color passive edges.}
We now show that, in phase $i$ of \Cref{phase-4}, in each $G^i$, it holds that passive edges $e$ have slack $\ge 1$, even if the colors used by edges that became passive in larger phases are removed from $L_e$.
In fact, let $e$ be an edge that became passive exactly during phase $i$. Its slack was at least $S / (1+\varepsilon)^{2i-2} \ge S/(1+\varepsilon)^{2k} \ge 1$. Each neighboring edge of $e$ in $G_i$ that got colored either in \Cref{phase-3} or in \Cref{phase-4} in some phase $j > i$, removes at most one color from the list of available colors of $e$, but also decreases by $1$ the degree of $e$ in the graph induced by passive edge of $G^i$, and hence the slack is still at least $1$.

\paragraph{The final graphs have constant maximum degree.}
Consider an arbitrary graph $G^{k+1} \in \mathcal{G}^{k+1}$. Since the slack is $\ge 1$, the maximum degree is upper bounded by the maximum list size of each edge, which in turn is upper bounded by the largest color space. At each step, by \Cref{lem:splithighdegree}, if we start with a color space of size $p$, we obtain a new color space of size at most $\lceil p / 2 \rceil$. Let $C'$ be the smallest power of $2$ that is larger than $C$. Then, 
	\[
		\deg_{G^{k+1}}(e) \le \frac{C'}{2^k} \le \frac{2 C}{2^k} \le  4.
	\]
	Hence, the obtained maximum degree is constant, and since the slack is at least $1$, that is, each edge satisfies $|L_e| > \deg(e)$, and since the graph is $2$-vertex colored, then we can assign to each edge of $G_i$ a color from its list in constant time, such that neighboring edges have different colors.
\paragraph{The graphs induced by passive edges have small maximum degree.}
	If in \Cref{phase-2} an edge becomes passive, then its degree in $G_i$ is strictly less than $\beta / \varepsilon$, and hence also the maximum degree in the graph induced by passive edges of $G_i$ is bounded by $\beta/\varepsilon$. Hence, in \Cref{phase-4}, we color, in parallel, graphs of maximum degree at most $\beta / \varepsilon$, and this can be done in $O(\beta / \varepsilon)$.
	
\paragraph{A bound on the running time.}
We now provide an upper bound on the running time of our algorithm. In \Cref{phase-2} we perform $k = O(\log C)$ phases, and each phase costs $O\big(\frac{\log^4 \Delta}{\varepsilon^6}\big)$ rounds, for $\varepsilon = 1 / \log C$. Then, in \Cref{phase-3}, we color graphs of constant maximum degree, which costs $O(1)$ rounds. Finally, in \Cref{phase-4}, we perform $k = O(\log C)$ phases, and each phase costs $O(\beta / \varepsilon)$ rounds.
In total, we spend $O\big(\log C \cdot \frac{\log^4 \Delta}{\varepsilon^6} + \log C \cdot \frac{\beta}{\varepsilon}\big) = O\log^7 C \cdot \log^4 \Delta)$ rounds.

\end{proof}

We also make use of the following lemma, that has been already used in previous works related to edge coloring. Essentially, this lemma allows us to increase the slack, from $1$ to an arbitrary constant.
\begin{lemma}[Lemma 4.2 of \cite{BalliuKO20}, rephrased]\label{lem:moreslack}
	For any $S > 1$, and any $k$, given an instance of $P(\bar{\Delta},1,C)$, it is possible to spend  $O(S^2 \log k) \cdot T(\bar{\Delta},S,C) + O(\log k \log^* X)$ rounds, and partially solve it such that the graph induced by uncolored edges has edge degree at most $\bar{\Delta}/k$, if an initial edge coloring with $X$ colors is given.
\end{lemma}

We are now ready to prove our main result. \Cref{thm:mainsimple} directly follows from the following theorem.
\begin{theorem}
  The $(\deg(e)+1)$-list edge coloring problem can be solved in
  $O(\log^7 C \cdot \log^{5} \Delta + \log^* n)$ deterministic rounds
  in the \LOCAL model, where $C$ is the size of the color space.
\end{theorem}
\begin{proof}
	At the beginning, given a graph $G$, we spend $O(\log^* n)$ rounds to compute an $O(\Delta^2)$-vertex coloring. Then, we recursively do the following.
	
	We compute a $\frac{\Delta}{2}$-defective vertex $c$-coloring, where $c = O(1)$, which can be done in $O(1)$ rounds given an $O(\Delta^2)$ coloring \cite{barenboim14distributed}. Let $G_\mathrm{B}$ be the graph induced by bichromatic edges of $G$, and $G_\mathrm{M}$ the graph induced by monochromatic edges of $G$.
	We partially color the edges of $G_\mathrm{B}$, as follows. We go through each possible pair of colors $(a,b)$, we consider the bipartite $2$-colored graph $G_{a,b}$ induced by the uncolored edges that are incident to nodes of different colors, one of color $a$ and the other of color $b$, and we apply \Cref{lem:moreslack} combined with \Cref{lem:solveslack} with parameters $S = e^2$, $C = \bar{\Delta}+1$, $k = 16 c$, and for each edge we define its list of available colors as $L_e \setminus \{c_e' ~|~ e' \in E_G \land e \cap e' \neq \emptyset \land e' \text{ is colored with color } c_e'\}$. We obtain that the graph induced by uncolored edges of $G_{a,b}$ has edge degree at most $(2\Delta-2) / k \le \Delta / (8 c)$, and hence also the maximum vertex degree is bounded by $\Delta/ (8c) + 1 \le \Delta / (4c)$. 
	
	Consider now an arbitrary node of color $\hat{c}$ of $G_\mathrm{B}$. It is part of at most $c$ graphs $G_{\hat{c},\cdot}$, and in each of them the graph induded by uncolored edges has maximum degree $\Delta / (4 c)$. Hence, the graph induced by uncolored edges of $G_\mathrm{B}$ has degree at most $\Delta/4$.
	
	Hence, we obtain that some edges remain uncolored because they are part of $G_\mathrm{M}$, which has degree at most $\Delta/2$, and some edges remain uncolored because the coloring in $G_\mathrm{B}$ is partial, and there the maximum degree in the graph induced by uncolored edges is at most $\Delta/4$. Hence, in $G$, the graph induced by uncolored edges has maximum degree at most $3\Delta/4$.  Note that we preserve the invariant that the slack is at least $1$, since for every edge it holds that, if a neighbor gets colored, then its degree decreases by $1$, and the size of the list of available colors decreases by at most $1$. 
	
	After $O(\log \Delta)$ steps of recursion, the obtained graph has constant maximum degree, and hence we can solve the remaining instance in $O(\log^* n)$ rounds.
	
	Since $k = O(1)$, by \Cref{lem:moreslack} and \Cref{lem:solveslack}, each step costs $O(\log^7 C \cdot \log^{4} \Delta)$ rounds, and hence we spend $O(\log^7 C \cdot \log^{5} \Delta + \log^* n)$ time in total.
\end{proof}

\end{document}